\documentclass[a4paper,12pt]{article}
\pdfoutput=1
\usepackage{jheppub}
\usepackage[T1]{fontenc}
\usepackage{amsthm}
\graphicspath{{figures/}}
\makeatletter
\renewcommand{\@fpheader}{}
\makeatother
\newtheorem{proposition}{Proposition}[section]
\newtheorem{theorem}[proposition]{Theorem}
\newtheorem{corollary}[proposition]{Corollary}
\renewcommand{\beforetochook}{\pagestyle{myplain}\pagenumbering{roman}\small}
\renewcommand{\afterTocSpace}{\normalsize\bigskip\medskip}
\renewcommand{\afterTocRuleSpace}{\clearpage}
\hypersetup{pdftitle={Relativistic elastic shells: material support and cavity geometry},
 pdfauthor={An T. Le},pdflang={en-US},bookmarksdepth=2}

\title{Relativistic elastic shells: material support and cavity geometry}
\author[a,b,c]{An T. Le}
\affiliation[a]{Center for Environmental Intelligence, VinUniversity, Hanoi, Vietnam}
\affiliation[b]{College of Engineering and Computer Sciences, VinUniversity, Hanoi, Vietnam}
\affiliation[c]{Intelligent Autonomous Systems, TU Darmstadt, Germany}
\emailAdd{an.lt@vinuni.edu.vn}
\emailAdd{an@robot-learning.de}

\abstract{
We derive the source moments that determine the leading tidal field and
rotational clock shift inside a finite-thickness relativistic elastic shell. In weak
gravity and slow rotation, these moments depend on the elastic response
of the material and its two free vacuum faces. Two choices of elastic
moduli, each causal near relaxation, give opposite signs for both observables.
Both shells have oblate faces and gain rotational mass. The tide depends on
differential flattening, bulk deformation, and rotational stress and energy;
the central clock shift measures a distinct inverse-radius moment.
A single constitutive action determines the material support and response.
Nonnegative isotropic pressure cannot support a static spherical shell
of positive-density matter with two free vacuum faces, whereas tangential
elastic stress admits equilibria near relaxation. For sufficiently weak
gravity these equilibria obey the dominant energy condition and have
subluminal sound speeds. Their static cavities are locally flat, with
clocks redshifted relative to infinity. At leading weak-gravity order,
static redshift and linear frame dragging probe the same inverse-radius
mass moment. Energy estimates
establish radial linear stability and exclude exponentially growing
odd-parity modes for weakly gravitating spherical equilibria. Nonlinear
compression can nevertheless produce unbounded gradients while the material
remains causal. Nonspherical even-parity and rotating stability, together
with quantitative bounds at finite rotation, remain open. The cavity
observables distinguish material responses that surface shape and total
mass alone do not determine.
}

\begin{document}
\maketitle

\section{Introduction}
\label{sec:intro}

Inside a static spherical shell, a clock can remain at rest without support
while ticking more slowly than a clock at infinity. The regular vacuum
cavity is locally flat; the redshift compares clocks in different parts of
the same spacetime. Slow rotation adds a relative rotation of the cavity's
inertial frame, followed at second order by tidal curvature that changes
the separation of freely falling tracers. We ask how the matter surrounding
the cavity determines these measurements.

The support problem comes first. Gravity pulls the shell inward, but neither
vacuum face exerts a force on the material. A static spherical shell cannot
meet these conditions with nonnegative isotropic pressure alone. We use a
finite-thickness elastic annulus, whose tangential compression supplies the
missing support. One elastic action determines both the equilibrium stress
and its response to rotation. The cavity contains only test clocks and
freely falling tracers; their stress-energy is neglected.

Our main result is an explicit relation between the elastic response and
two cavity observables: the leading tide and the rotational correction to
the central clock rate. Their source moments probe different features of
the material. In weak gravity, the static redshift and linear frame dragging
weight each spherical mass layer by its inverse radius. The leading cavity tide is more sensitive
to matter near the inner face. It depends on how the two faces flatten
relative to one another, as well as on deformation and stress within the
material. Two choices of elastic moduli, each causal near relaxation,
give oblate faces, flattened
along the rotation axis, yet interchange the stretching and focusing
directions of nearby tracers.
The shape of either surface alone therefore does not determine the cavity
geometry. Rotation also changes the central clock rate. Its correction
depends on another source moment, so an increase in total mass need not
make that clock run more slowly.

Finite-thickness elastic shells and their two free boundaries were studied
by Losert-Valiente Kroon~\cite{losert2006}, including a Newtonian existence
argument. Brito, Carot and Vaz~\cite{brito2020} constructed relativistic
elastic thick shells with causal propagation and dominant-energy bounds,
matched to cosmological regions. Here both adjacent regions are vacuum,
and the exterior is asymptotically flat. Einstein--Vlasov shells can also
have vacuum
centers~\cite{rein1999}, and Einstein clusters provide a circular-orbit
support mechanism~\cite{boehmer2007}. We use the material framework of
relativistic elasticity~\cite{karlovini2003,alho2024} and established static and
rotating elastic-body theory~\cite{andersson2008,andersson2010}.
The new calculation identifies the source moments for the cavity tide and
clock shift and evaluates them for two materials with opposite responses,
both causal near relaxation. Static shell existence and the rotating elastic-body framework
provide the foundations; the rotating existence result is an application
of that established theory.

The connection to subluminal warp shells is the geometry of a moving body
with a vacuum interior~\cite{fuchs2024constvel}. Recent analyses emphasize
the distinction between coordinate motion and material motion, and between
a prescribed stress tensor and a matter model with its own equations of
motion~\cite{barzegar2026classification,buchert2026realizations}.
Here the action specifies those equations and the stresses at both faces.
The solutions describe gravitating elastic bodies with cavities. Their
nonconstant lapse and curved spatial slices lie outside the unit-lapse,
spatially flat ansatz of many warp-metric restrictions~\cite{lobo2004,santiago2021};
Olum's superluminal-travel criterion~\cite{olum1998} has different hypotheses.
Neither propulsion nor a travel-time advantage follows from the cavity.

Sections~\ref{sec:isotropic}--\ref{sec:centrifugal} develop this sequence:
material support, static redshift, frame dragging, and the changes in
curvature and clock rate caused by rotation.
Section~\ref{sec:stability-summary} separates linear stability from
nonlinear smooth evolution and discusses uniform motion and additional
loads. The appendices give the coefficient derivations, dynamical proofs,
and exact transformation to a uniformly moving shell.

The signature is $(-+++)$,
$R^a{}_{bcd}=\partial_c\Gamma^a_{db}-\partial_d\Gamma^a_{cb}
+\Gamma^a_{ce}\Gamma^e_{db}-\Gamma^a_{de}\Gamma^e_{cb}$, and the
cosmological constant vanishes. We use $c=1$ and usually $G=1$, restoring
$G$ in existence and dynamical arguments. The mass function $m$, total
mass $M$, angular momentum $J$, and four-momentum $P^\mu$ are geometrized;
when $G$ is explicit, their physical values are $m/G$, $M/G$, $J/G$, and
$P^\mu/G$.

\section{Free boundaries and material support}
\label{sec:isotropic}

Let $r$ be the areal radius, so that each sphere has area $4\pi r^2$.
Write the static spherical metric as
\begin{equation}
 ds^2=-e^{2\Phi(r)}dt^2+\frac{dr^2}{1-2m(r)/r}+r^2d\Omega^2,
 \qquad F(r)=1-\frac{2m(r)}r>0.
 \label{eq:spherical}
\end{equation}
For $T^{\hat a}{}_{\hat b}=\operatorname{diag}(-\rho,p_r,p_t,p_t)$,
the independent equations can be written
\begin{align}
 m'&=4\pi r^2\rho,&
 \Phi'&=\frac{m+4\pi r^3p_r}{r(r-2m)},\label{eq:spherical-einstein}\\
 p_r'&=-(\rho+p_r)\Phi'+\frac{2(p_t-p_r)}r.
 &&\label{eq:anisotropic-tov}
\end{align}
Here $\rho$ is the energy density, $p_r$ the radial pressure, and $p_t$ the
tangential pressure, all measured in the material rest frame. The function
$m(r)$ is the enclosed geometrized mass; $e^\Phi$ sets the rate of a
stationary clock relative to the time coordinate $t$.
The last equation is radial stress conservation. Together with the first two,
it supplies the angular Einstein equation in the smooth bulk.
For a shell occupying $a\le r\le b$, with an empty regular cavity,
$m(a)=0$. Matching to vacuum without a surface layer requires
$p_r(a)=p_r(b)=0$, with continuous induced metric and extrinsic curvature
on each timelike boundary~\cite{israel1966}. The induced metric gives
distances and times along the boundary; the extrinsic curvature describes
how it is embedded in spacetime. Continuity of both is called Darmois
matching and excludes an additional sheet of stress-energy on either face.
For the metric above, the nontrivial mixed extrinsic curvatures of a static
face, with the same increasing-$r$ normal on its two sides, are
$\mathcal K^t{}_t=\sqrt F\,\Phi'$ and
$\mathcal K^\theta{}_\theta=\mathcal K^\phi{}_\phi=\sqrt F/r$.
Continuity of $m$ and $e^\Phi$, together with $p_r=0$, makes these equal
to their vacuum values. The density itself may jump: it enters $m'$, but
does not create a surface layer when these matching conditions hold.

The support mechanism is visible directly in stress conservation:
\begin{equation}
 -p_r'-(\rho+p_r)\Phi'+\frac{2(p_t-p_r)}r=0,\qquad
 \int_a^b\frac{2(p_t-p_r)}r\,dr
   =\int_a^b(\rho+p_r)\Phi'\,dr.
 \label{eq:integrated-support}
\end{equation}
Multiplying the local balance by $\sqrt F$ gives proper radial force
densities. The integral identity follows from $p_r(a)=p_r(b)=0$. If the gravitational
term on the right is positive, tangential pressure must exceed radial
pressure somewhere. Here positive $p_i$ denotes compression. Its anisotropic
contribution can point outward even though gravity points inward; a
traction-free face requires $p_r=0$, not $p_t=0$.

\begin{proposition}[Isotropic inner-boundary obstruction]
\label{prop:isotropic}
Let $0<a<b$, and let a horizon-free static spherical shell have continuous,
nonnegative density and pressure, isotropic stress, and nonzero density
somewhere in $(a,b)$. Assume the bulk Einstein equations, differentiability
there, $m(a)=0$, and vacuum matching without a surface layer at $a$.
These assumptions are incompatible.
\end{proposition}
\begin{proof}
Equations~\eqref{eq:spherical-einstein}--\eqref{eq:anisotropic-tov} give
$m\ge0$ and
\begin{equation}
 p'= -\frac{(\rho+p)(m+4\pi r^3p)}{r(r-2m)}\le0.
 \label{eq:isotropic-tov}
\end{equation}
Since $p(a)=0$ and $p\ge0$, monotonicity forces $p=0$ throughout the shell.
Continuity and nontriviality of $\rho$ provide an interior point with both
$\rho>0$ and $m>0$. Equation~\eqref{eq:isotropic-tov} then gives $p'<0$,
contradicting $p=0$.
\end{proof}
Tangential stress can supply the required support while obeying the
dominant energy condition; Appendix~\ref{sec:exact} gives an explicit example
with $p_r=0$. To determine how a shell responds when its load changes,
we need a law relating stress to deformation.

\subsection{A material law for support and response}

Material coordinates $X^A$ are labels carried by the particles. A fixed
Euclidean material metric $k_{AB}$ records their separations in the relaxed,
unstressed body. The spacetime metric gives their actual proper separations.
In the three principal directions, let $n_i^{-1}$ be the stretch relative
to relaxation. Thus $n_i=1$ at relaxation and compression increases $n_i$;
the product $n=n_1n_2n_3$ is the relative particle number density.

These quantities have a covariant definition. Set
$H^{AB}=g^{ab}\partial_aX^A\partial_bX^B$; the $n_i^2$ are the eigenvalues
of $H^{AC}k_{CB}$. The material map has rank three and a timelike kernel,
spanned by the future unit material velocity $u^a$, so $u^a\partial_aX^A=0$.
We choose a compressional energy and the quasi-Hookean shear invariant of
Ref.~\cite{karlovini2003}:
\begin{align}
 s^2&=\frac1{12}\sum_{i<j}\left(\frac{n_i}{n_j}-\frac{n_j}{n_i}\right)^2,
 \nonumber\\
 \rho(n_i)&=m_0n+\frac{\kappa}{2}(n-1)^2+\mu_0 n s^2,
 \qquad p_i=n_i\frac{\partial\rho}{\partial n_i}-\rho.
 \label{eq:elastic-law}
\end{align}
The three terms in $\rho$ account for rest mass, volume change, and shear.
The constant $m_0$ is the relaxed energy density, while $\kappa$ and $\mu_0$
are the bulk and shear moduli at relaxation. All three have units of energy
density; away from relaxation the shear coefficient is $\mu_0n$.
The material action is
$S_m=-\int_{\mathcal B}\rho\sqrt{-g}\,d^4x$ over the occupied material
worldtube $\mathcal B$; its exterior is vacuum. Its variations give
\begin{equation}
 T_{ab}=2\rho_{H^{AB}}\partial_aX^A\partial_bX^B-\rho g_{ab},\qquad
 \nabla_a\!\left(2\rho_{H^{AB}}\nabla^aX^B\right)=0
 \label{eq:elastic-source}
\end{equation}
in Cartesian material coordinates, together with zero physical traction on
each free face. Since $u^a\partial_aX^A=0$, the tensor obeys
$T_{ab}u^b=-\rho u_a$; its spatial eigenvalues are the $p_i$ in
Eq.~\eqref{eq:elastic-law}. The law is used only near relaxation. It describes
a homogeneous, isotropic, perfectly elastic solid: there is no plasticity,
viscosity, or fracture criterion. Material isotropy means that the relaxed
body has no preferred local direction. A deformation can still produce
unequal principal pressures, including the difference between $p_r$ and
$p_t$ needed for shell support.
There is no separate surface energy or surface tension in this action.
The solid can have nonzero density at a free face. In the distributional
conservation law, extending its stress by zero into vacuum produces a
surface term proportional to $T^{ab}n_a$; zero traction removes that term.
Thus a sharp material boundary is consistent with stress conservation,
without prescribing an additional transition layer.
The state $n_i=1$ has zero stress. Its transverse and longitudinal
squared characteristic speeds are respectively
$\mu_0/m_0$ and $(\kappa+4\mu_0/3)/m_0$.
Transverse waves probe shear, while longitudinal waves also compress the
material. Positive shear stiffness allows neighboring radial layers to
support one another while both faces remain free.

The stress tensor has a material rest frame and three principal stresses,
so it is of Hawking--Ellis Type I~\cite{martinmoruno2018core}.
Its rest-frame energy density is $\rho$, and its principal pressures are
$p_1,p_2,p_3$. The null (NEC), weak (WEC), and dominant (DEC) energy
conditions therefore reduce to
\begin{align}
 \mathrm{NEC}&\iff \rho+p_i\ge0\quad\hbox{for each }i,\nonumber\\
 \mathrm{WEC}&\iff \rho\ge0\ \hbox{and NEC},\nonumber\\
 \mathrm{DEC}&\iff \rho\ge |p_i|\quad\hbox{for each }i.
 \label{eq:ec}
\end{align}
These rest-frame inequalities characterize the energy conditions for all
causal observers. In particular, the DEC requires nonnegative energy density
and a causal energy flux for every timelike observer.

\section{Elastic equilibrium, mass, and cavity redshift}
\label{sec:elastic}

We distinguish material radii $A\le R\le B$ in the relaxed reference
annulus from physical areal radii $a\le r\le b$ in the equilibrium shell.
The strain determines the map between them. The dimensionless static
parameters are
\begin{equation}
 \beta=B/A,\qquad \chi=Gm_0A^2,\qquad
 k=\kappa/m_0,\qquad \mu=\mu_0/m_0.
 \label{eq:dimensionless-parameters}
\end{equation}
At fixed $(\beta,k,\mu)$, weak gravity means $\chi\to0$.
Rotation introduces $\varepsilon_\Omega=(A\varpi)^2$, where $\varpi$ is
the angular velocity measured relative to time at infinity.
These parameters separate geometry, gravitational strength, and material
stiffness; the individual choices of length and density units do not select
a special shell. The examples below use $(\beta,k,\mu)=(2,1/5,1/10)$.
Their relaxed sound speeds are $c_T=1/\sqrt{10}$ and $c_L=1/\sqrt3$,
so this is a relativistically stiff ideal material.
Weak gravity at fixed $(k,\mu)$ is therefore not a nonrelativistic
material limit. In particular, the leading rotational calculation retains
special-relativistic strain and stress contributions even as $\chi\to0$.
All small-parameter limits below keep $(\beta,k,\mu)$ fixed, with positive
moduli. Slow rotation must also keep the strain small. Softer materials
deform more at a given rim speed $B|\varpi|$, so a small speed alone is
not an error bound for the deformation expansion.

For a spherical material shell $A\le R\le B$, put
$\nu=n_r=\sqrt F\,dR/dr$ and $\eta=n_t=R/r$. In material radius the full
static equations are
\begin{equation}
 \begin{gathered}
 \begin{aligned}
 r_R&=\frac{\sqrt F}{\nu}, &
 m_R&=4\pi r^2\rho\,r_R,\\
 \eta_R&=\frac1r-\frac\eta r r_R, &
 \Phi_R&=\frac{m+4\pi r^3p_r}{r(r-2m)}r_R,
 \end{aligned}\\
 \nu_R=\frac{[-(\rho+p_r)\Phi_r+2(p_t-p_r)/r]r_R
                    -(\partial_\eta p_r)\eta_R}{\partial_\nu p_r}.
 \end{gathered}
 \label{eq:elastic-ode}
\end{equation}
Here $p_r,p_t$ are evaluated from Eq.~\eqref{eq:elastic-law} at
$(n_1,n_2,n_3)=(\nu,\eta,\eta)$. The boundary conditions are
\begin{equation}
 m(A)=0,\quad p_r(A)=p_r(B)=0,\quad
 \Phi(B)=\tfrac12\log[1-2m(B)/r(B)].
 \label{eq:elastic-bc}
\end{equation}
The material strain determines the physical radii.

\begin{samepage}
\begin{theorem}[Nonlinear spherical elastic shells]
\label{thm:elastic-existence}
Fix $0<A<B$ and positive $m_0,\kappa,\mu_0$ with
$\kappa+4\mu_0/3<m_0$. For all sufficiently small $G>0$, the material
law~\eqref{eq:elastic-law} admits a locally unique spherical static shell
near $r=R$, with a flat vacuum cavity, a Schwarzschild exterior, and zero
traction at both faces. The material satisfies strict DEC and has real,
strictly subluminal characteristic speeds.
\end{theorem}
\end{samepage}
\begin{proof}
Restore $G$ by replacing $4\pi\rho$ and $4\pi p_r$ in the mass and lapse
equations by $4\pi G\rho$ and $4\pi Gp_r$. Choose the inner physical radius
$r(A)=q$ and integrate outward from $m(A)=0$.
Since $\partial_\nu p_r=\kappa+4\mu_0/3>0$ at relaxation, the inner
traction condition determines $\nu(A)$ smoothly from $q$ near $A$.
The bulk equations are smooth on a neighborhood of the relaxed solution
throughout $[A,B]$. Let $\mathcal R(q,G)$ be the resulting outer traction;
the remaining boundary condition is $\mathcal R(q,G)=0$.

At $G=0$, the linear displacement $u=\partial r/\partial q$ solves
$u''+2u'/R-2u/R^2=0$, $u(A)=1$, with zero inner traction. Hence
\begin{equation}
 u(R)=\frac{4\mu_0 R/A+3\kappa A^2/R^2}{4\mu_0+3\kappa},\qquad
 \partial_q\mathcal R(A,0)=
 -\frac{12\kappa\mu_0}{A(4\mu_0+3\kappa)}
       \left(1-\frac{A^3}{B^3}\right)<0.
 \label{eq:shooting-derivative}
\end{equation}
The implicit function theorem gives $q(G)$ and local uniqueness in the
spherical class; exterior normalization fixes the lapse constant.
For small positive $G$, the density remains positive, $F>0$, and the
material map remains invertible. The free-traction conditions give Darmois
matching. Strict DEC and the characteristic gaps hold at relaxation and
persist uniformly on the compact material interval and sphere of wave
directions. More explicitly, the kinetic acoustic matrix is positive at
relaxation, and the symmetric spatial acoustic matrix lies strictly between
zero and the kinetic matrix for every unit wave covector. These strict
matrix inequalities persist under small strain, including at repeated
principal densities; no differentiable choice of eigenvectors is needed.
\end{proof}
This spherical construction complements the Newtonian shell
construction~\cite{losert2006} and general relativistic static-body
theory~\cite{andersson2008}.
Its small-coupling threshold is not quantified.
The zero-shear limit is singular for this construction: the shooting
derivative in Eq.~\eqref{eq:shooting-derivative} vanishes as $\mu_0\to0$.
The theorem therefore cannot be extended to a perfect fluid by
setting the shear modulus to zero.

\subsection{A weakly gravitating example}

For $A=10$, $B=20$, $m_0=10^{-6}$, $\kappa=m_0/5$, and $\mu_0=m_0/10$,
with $G=1$ in a chosen geometrized length unit, shooting gives
\begin{equation}
 r(A)\simeq9.98529,\quad r(B)\simeq19.98318,\quad M\simeq0.0292968.
 \label{eq:elastic-example}
\end{equation}
The computed profile has $n_r,n_t\in[0.99,1.01]$, the constitutive
domain treated in Appendix~\ref{app:material-bounds}. Its smallest computed
DEC margin $\rho-\max(|p_r|,|p_t|)$ is about $1.0007m_0$.
The decimal values illustrate the equilibrium; the existence proof above
is independent of this calculation.
Figure~\ref{fig:design} shows how this shell is supported. Near the
inner face, radial pressure rises outward, so its gradient force points
inward. Tangential compression supplies the outward term in
Eq.~\eqref{eq:integrated-support}; farther out, the pressure gradient
reverses and also helps oppose gravity.
Stress depends on all three strains: radial extension can coexist with
positive radial pressure because the tangential directions are compressed.

The density has a finite jump to vacuum at each free surface. The traction
conditions give Darmois matching without a surface stress; this elastic
example is piecewise smooth. The prescribed density in
Appendix~\ref{sec:exact} instead vanishes smoothly at the faces.

Appendix~\ref{app:material-bounds} derives the acoustic matrices from the
action and proves the uniform bounds
\begin{equation}
 \frac1{16}\le c_{\rm sound}^2\le\frac25,
 \qquad \rho-\max_i|p_i|>\frac{24}{25}m_0,
 \qquad n_i\in[0.99,1.01],
 \label{eq:material-box-bounds}
\end{equation}
for the selected material ratios.
The sound speeds are measured in the local material rest frame.
These conservative bounds hold for every state in the displayed strain
range. They apply to an equilibrium as long as its strain remains in that
range. They constrain local signal propagation; the growth of perturbations
of the whole gravitating shell is a separate question, treated in
Section~\ref{sec:stability-summary}.

To give the example a physical scale, choose a length unit $L_0$.
A quoted mass $\widehat M$ then corresponds to
$M_{\rm phys}=(c^2L_0/G_N)\widehat M$; energy density and stress scale as
$c^4/(G_NL_0^2)$, and angular frequency as $c/L_0$.
Weak compactness alone does not fix a laboratory mass or stress scale.
The chosen elastic action has not been identified with a laboratory material.

\subsection{The cavity clock and the mass at infinity}
\label{sec:mass-clock}

With the equilibrium fixed, the lapse determines the clock comparison.
We keep $G$ explicit in this subsection and write $e^{2\Lambda}=F^{-1}$.
In the regular cavity, $m=0$ and $\Phi=\Phi_c$ is constant. Rescaling time
there by $T=e^{\Phi_c}t$ gives the Minkowski metric. The rate relative to a
clock at infinity is fixed by matching through the shell:
\begin{equation}
 \alpha_c=e^{\Phi_c}=\sqrt{1-2M/b}\,
 \exp\!\left[-\int_a^b\frac{m+4\pi G r^3p_r}{r(r-2m)}\,dr\right].
 \label{eq:cavity-clock-rate}
\end{equation}
For the numerical example, $\alpha_c\simeq0.99811$. A stationary cavity clock has
$d\tau=\alpha_c\,dt$, and light of locally measured frequency $\nu_c$
received by an observer at rest at infinity has $\nu_\infty=\alpha_c\nu_c$,
because the photon energy associated with the stationary time symmetry is
conserved along the ray. The static observer's
proper acceleration is $\sqrt F\,\Phi'$ and vanishes throughout the cavity.
A local experiment can detect no static tidal field in this cavity,
whereas comparison with a distant clock detects the redshift. This is a
property of the isolated shell; an imposed external field would require a
different solution.

The pressure that supports the shell also contributes to the mass measured
at infinity. To relate that mass to the material energy, the spherical
field equations give the exact identity
\begin{equation}
 \left(r^2e^{\Phi-\Lambda}\Phi'\right)'
 =4\pi G r^2e^{\Phi+\Lambda}(\rho+p_r+2p_t).
 \label{eq:komar-derivative}
\end{equation}
It follows by adding the orthonormal time and spatial Einstein equations,
or by differentiating the left side and substituting
Eq.~\eqref{eq:spherical-einstein} and stress conservation. Its inner value
vanishes and its outer value is $M$. Therefore
\begin{equation}
 \frac MG=4\pi\int_a^b e^{\Phi+\Lambda}
                  (\rho+p_r+2p_t)r^2\,dr.
 \label{eq:komar-mass}
\end{equation}
This is the Komar mass, defined by the stationary time symmetry. It equals
the Arnowitt--Deser--Misner (ADM) mass at spatial infinity. In equilibrium
it also equals $M/G=4\pi\int_a^b\rho r^2dr$.
The pressure and redshift factors in the Komar integral therefore cannot
be varied independently.

The proper energy $\mathcal M_{\rm prop}=4\pi\int_a^b\rho r^2dr/\sqrt F$
instead sums the energies measured in local material rest frames.
Write $\epsilon=\rho/n$ for energy per relaxed material volume. Since
$n=\sqrt F\,R^2/(r^2r_R)$, the same integrals become
\begin{equation}
 \frac MG=4\pi\int_A^B R^2\epsilon\sqrt F\,dR,
 \qquad \mathcal M_{\rm prop}=4\pi\int_A^B R^2\epsilon\,dR.
 \label{eq:material-mass-identities}
\end{equation}
Thus $\mathcal M_{\rm prop}-M/G>0$ for a nontrivial positive-density
static shell. The ADM mass already includes gravitational binding.
The reference rest energy $\mathcal M_0=4\pi m_0(B^3-A^3)/3$ is fixed by
material content; $\mathcal M_{\rm prop}$ also includes strain energy.

On the weak-coupling branch, strain is $O(G)$ and the relaxed specific
energy is stationary, so $\epsilon=m_0+O(G^2)$. With
$\mathfrak m_0(R)=4\pi m_0(R^3-A^3)/3$,
expanding Eq.~\eqref{eq:material-mass-identities} gives
\begin{align}
 \frac MG&=\mathcal M_0-G\mathcal B_0+O(G^2),\nonumber\\
 \mathcal B_0&=4\pi m_0\int_A^B R\mathfrak m_0(R)\,dR
 =\frac{16\pi^2m_0^2}{3}
   \left[\frac{B^5-A^5}{5}-\frac{A^3(B^2-A^2)}2\right]>0.
 \label{eq:weak-binding}
\end{align}
The coefficient is the Newtonian binding energy of the relaxed annulus.
The error statement is at fixed annulus and moduli; it gives no bound at
arbitrary compactness.

\begin{figure}[tbp]
 \centering
 \includegraphics[width=\linewidth]{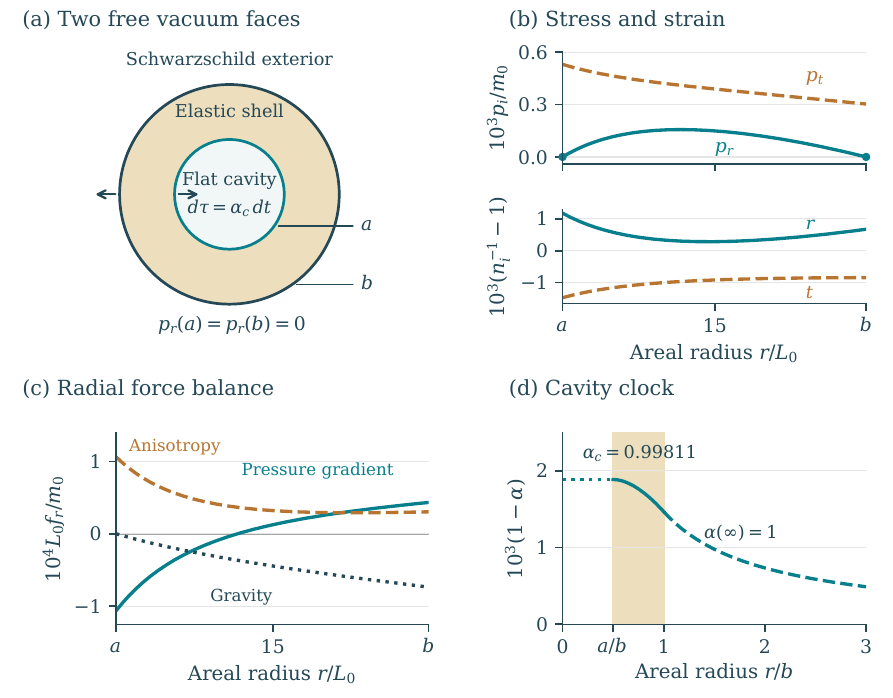}
 \caption{Static support for the elastic shell of Eq.~\eqref{eq:elastic-example}.
 (a) Areal-coordinate section with outward material normals at both free
 faces. The flat cavity has lapse $\alpha_c$ relative to infinity.
 (b) Radial and tangential pressures (positive in compression) and departures
 of the proper stretches $1/n_r=d\ell/dR$ and $1/n_t=r/R$ from unity.
 Free radial traction permits nonzero tangential stress.
 (c) Pressure-gradient, gravity and anisotropy terms of
 Eq.~\eqref{eq:integrated-support}, per unit areal radius. Outward is
 positive; their sum vanishes. Anisotropy supplies outward support near $a$.
 (d) Proper clock rate $\alpha=d\tau/dt$, with $t$ normalized at infinity.
 Shading marks the material.}
 \label{fig:design}
\end{figure}

\section{Rotation and inertial-frame dragging}
\label{sec:circulation}

Rotation changes the comparison between the cavity's inertial axes and
those at infinity. On the branch considered here, reversing the angular
velocity reverses frame dragging and leaves the shape unchanged. Dragging
therefore begins at first order, while deformation begins at second order.
We treat the dragging here and the deformation in Section~\ref{sec:centrifugal}.

An isolated shell whose particles follow its globally timelike stationary
symmetry is static under the hypotheses of Appendix~\ref{app:stationary},
even if its shape is not spherical. Internal motion is therefore necessary
for the stationary dragging studied here. The appendix also restricts
stationary label motion near relaxation to Euclidean rotations, but that
kinematic result does not establish rigid rotation in physical space.
For rigid rotation we use the following existence result.

In rigid rotation all particles have the same angular velocity, and their
strain is constant along their worldlines. The equilibrium shape can still
depend on that angular velocity.

\begin{samepage}
\begin{theorem}[Slowly rotating elastic shells~\cite{andersson2010}]
\label{thm:rotating}
Fix the law~\eqref{eq:elastic-law} with positive moduli and
$\kappa+4\mu_0/3<m_0$, and a round relaxed material annulus $A<|X|<B$.
For sufficiently small $G>0$ and $|\varpi|$, there is a stationary,
axisymmetric Einstein--elastic shell near relaxation with regular vacuum
cavity and asymptotically flat exterior. Both faces have zero traction, and
\begin{equation}
 u\parallel K+\varpi L,\qquad [K,L]=0,
 \label{eq:rigid-rotation}
\end{equation}
where $K$ is normalized at infinity and $L$ has $2\pi$-periodic axial orbits.
The matter obeys strict DEC and has causal material characteristics. The
solution is locally unique within the rigid-rotation construction after
fixing its gauge and rigid motions.
\end{theorem}
\end{samepage}
The annular-domain and constitutive hypotheses, including smooth
parameter dependence in material coordinates, are given in Appendix~\ref{app:rotating-existence}.

The cavity remains regular vacuum but need not remain flat; the rotating
exterior is generally not Schwarzschild. The existence neighborhood is
unquantified and need not contain the numerical shell at its stated
parameters. The first-order response about that sphere can nevertheless
be calculated without assuming a finite rotating continuation.

The linear stationary rotational equation about a spherical equilibrium is
more explicit. With $\varpi$ the angular velocity and $\omega(r)=O(\varpi)$,
\begin{equation}
 ds^2=-e^{2\Phi}dt^2+e^{2\Lambda}dr^2+r^2d\Omega^2
       -2r^2\sin^2\theta\,\omega(r)\,dt\,d\phi+O(\varpi^2).
 \label{eq:slow-rotation}
\end{equation}
The axial momentum equation is
\begin{equation}
 (r^4j\omega')'=-16\pi G r^4e^{\Lambda-\Phi}
                          (\rho+p_t)(\varpi-\omega),
 \qquad j=e^{-\Phi-\Lambda}.
 \label{eq:dragging}
\end{equation}
To obtain this equation, the shift in the convention of
Appendix~\ref{sec:source} is $\beta^\phi=-\omega$. Its only independent
nonzero extrinsic-curvature component is
\[
 K_{r\phi}=-\tfrac12 r^2\sin^2\theta\,e^{-\Phi}\omega'.
\]
Its trace vanishes. The material momentum density is
\[
 S_\phi=(\rho+p_t)r^2\sin^2\theta\,e^{-\Phi}(\varpi-\omega).
\]
Substituting these expressions into $D_jK^j{}_{\phi}=8\pi G S_\phi$
gives Eq.~\eqref{eq:dragging}, including its sign and normalization.
The source contains the tangential enthalpy $\rho+p_t$: the material's
energy density and tangential pressure both contribute to rotational inertia.
Regularity gives $\omega'=0$ in the cavity; asymptotic nonrotation gives
$\omega=2J/r^3$ in the exterior at this order, with $J$ the geometrized
ADM angular momentum.

Let $f(a)=1$, $f'(a)=0$ solve
\begin{equation}
 (r^4jf')'=16\pi G r^4e^{\Lambda-\Phi}(\rho+p_t)f,
 \qquad C=f(b)+\tfrac b3 f'(b).
 \label{eq:dragging-normalization}
\end{equation}
Then $\varpi-\omega=\varpi f/C$ in the material and
\begin{equation}
 \frac{\omega_{\rm cav}}{\varpi}=1-\frac1C,\qquad
 \frac J\varpi=\frac{b^4f'(b)}{6C}>0.
 \label{eq:dragging-response}
\end{equation}
Indeed, the positive right-hand side gives $f'>0$ beyond the inner face and
$C>f(b)>1$. Hence $0<\omega<\varpi$ for $\varpi>0$: the cavity's inertial
frame turns in the same sense as the material, at a smaller angular rate.
Within the cavity, $T=\alpha_c t$ and
$\widetilde\phi=\phi-\omega_{\rm cav}t$ make the metric Minkowskian through
first order. Its relative rotation with respect to infinity is fixed by
matching through the shell. Local flatness at this order therefore coexists
with a measurable dragging response; curvature enters at the next order.
For the numerical spherical example,
\begin{equation}
 \omega_{\rm cav}/\varpi\simeq0.002515,\qquad
 J/\varpi\simeq5.186\,L_0^3.
 \label{eq:dragging-example}
\end{equation}
Figure~\ref{fig:rotation}(a) shows this linear response.

The weak-gravity limit connects dragging directly to the static redshift.
At fixed $(\beta,k,\mu)$ the relaxed density is $m_0$, the equilibrium
pressures are $O(\chi m_0)$, and $a=A+O(\chi A)$,
$b=B+O(\chi A)$. To leading order, the cavity potential and the dragging
equation reduce to
\begin{equation}
 \Phi_c=-4\pi Gm_0\int_A^B R\,dR+O(\chi^2),\qquad
 f'(r)=\frac{16\pi Gm_0}{5r^4}(r^5-A^5)+O(\chi^2/A).
 \label{eq:weak-cavity-integrals}
\end{equation}
The first expression is the Newtonian potential at any interior point of
a spherical annulus. The second follows by integrating
$(r^4f')'=16\pi Gm_0r^4$ with $f'(A)=0$.
Substitution into $C=f(B)+Bf'(B)/3$ gives
\begin{align}
 1-\alpha_c&=2\pi\chi(\beta^2-1)+O(\chi^2),\nonumber\\
 \frac{\omega_{\rm cav}}{\varpi}
 &=\frac{8\pi}{3}\chi(\beta^2-1)+O(\chi^2)
 =\frac43(1-\alpha_c)+O(\chi^2).
 \label{eq:weak-clock-dragging}
\end{align}
Here $\omega_{\rm cav}/\varpi$ denotes the linear rotational response;
both angular rates use time normalized at infinity.
For a relaxed mass layer $d\mathcal M_0=4\pi m_0R^2dR$, the common radial
weight is $G\,d\mathcal M_0/R$. Thus both observables measure an
inverse-radius mass moment.
This is the weak-gravity form of the shell-induced dragging discussed in
Ref.~\cite{brill1966}.
The factor $4/3$ comes from the rotational field equation and vacuum
matching; it is independent of the selected elastic moduli at this order.
It is not an exact relation at finite compactness. For $\chi=10^{-4}$
and $\beta=2$, the leading redshift and dragging fraction are
$0.001885$ and $0.002513$, respectively.

The dragging solution also determines the moment of inertia and the
leading rotational energy. Let $\mathcal M=M/G$ and $\mathcal J=J/G$ be
the physical mass and angular momentum.
Along the smooth branch at fixed material content,
$d\mathcal M=\varpi\,d\mathcal J$ (Appendix~\ref{app:first-law}).
Reversal of rotation makes $\mathcal M$ even and $\mathcal J$ odd. Hence
\begin{equation}
 \begin{gathered}
 \mathcal J=\mathcal I_G\varpi+O(\varpi^3),\qquad
 \mathcal M=\mathcal M_G+\tfrac12\mathcal I_G\varpi^2+O(\varpi^4),\\
 \mathcal I_G=\frac{8\pi}{3}\int_a^b
 r^4e^{\Lambda-\Phi}(\rho+p_t)\frac fC\,dr>0.
 \end{gathered}
 \label{eq:rotating-first-law}
\end{equation}
Here $\mathcal I_G$ is the moment of inertia measured from the asymptotic
angular momentum. The last identity follows by integrating
Eq.~\eqref{eq:dragging-normalization}.
At $G=0$, $\mathcal I_0=8\pi m_0(B^5-A^5)/15$.
Thus the dragging response also fixes the leading rotational energy on
the existence branch. The first-order calculation at the numerical example
does not, by itself, establish such a branch there.

\subsection{A circular clock comparison}

A difference in arrival times gives an operational measure of this
rotation. Send two light signals in opposite directions around a fixed
loop, and compare their returns on the emitter's own clock. To describe
the experiment, use coordinates adapted to the guide. Where their
stationary Killing field is timelike, write
\begin{equation}
 g=-e^{2U}(dt+\psi_i dx^i)^2+e^{-2U}h_{ij}dx^i dx^j,
 \qquad e^{2U}>0,\quad h>0.
 \label{eq:stationary-quotient}
\end{equation}
For two light signals sent in opposite directions around the same fixed
spatial loop $\mathcal C$, the future-directed null condition gives
\begin{equation}
 dt=-\psi_i dx^i+e^{-2U}\sqrt{h_{ij}dx^i dx^j},\qquad
 \Delta\tau=-2e^{U(x_0)}\oint_{\mathcal C}\psi_i dx^i,
 \label{eq:proper-sagnac}
\end{equation}
where $x_0$ is the emitter and stationary mirrors or guides maintain the
loop. The symmetric path-length term cancels between the two circuits.
The delay is invariant under $t\mapsto t+f(x)$ and vanishes for a guide
at rest in the static shell, where $\psi=0$.

For an equatorial circular guide at radius $r$, let $\Omega_g$ be its
angular velocity relative to infinity. In coordinates rotating with the guide,
$g_{t\phi}$ becomes $g_{t\phi}+\Omega_g g_{\phi\phi}$ and
$g_{tt}$ becomes $g_{tt}+2\Omega_g g_{t\phi}+\Omega_g^2g_{\phi\phi}$.
Equation~\eqref{eq:proper-sagnac} therefore gives
\begin{equation}
 \Delta\tau=
 \frac{4\pi(g_{t\phi}+\Omega_g g_{\phi\phi})}
 {\sqrt{-g_{tt}-2\Omega_g g_{t\phi}-\Omega_g^2g_{\phi\phi}}},
 \label{eq:circular-clock}
\end{equation}
provided $g_{\phi\phi}>0$ and the guide worldlines are timelike.
The sign compares the positive
azimuthal circuit with the negative one. To first order in $\varpi$ and
$\Omega_g$, this is $4\pi r^2e^{-\Phi}(\Omega_g-\omega)$.
A guide corotating with the material has $\Omega_g=\varpi$ and includes the
kinematic Sagnac
delay. A guide fixed relative to infinity has $\Omega_g=0$ and measures
the term proportional to $-\omega$. In the cavity these are ideal test
guides whose stress-energy is neglected. A guide following the local
inertial axes, $\Omega_g=\omega_{\rm cav}$, has zero delay at this order.

\section{Deformation, cavity tides, and the central clock}
\label{sec:centrifugal}

At first order, rotation drags the cavity's inertial frame without producing
tidal curvature there. At second order, deformation and stress determine
the leading tide and change the central clock rate. We calculate both
from the same material response.

\subsection{Centrifugal deformation}

The leading deformation in weak gravity follows from the elastic problem
at $G=0$. This is a special-relativistic calculation: distances between
neighboring material particles are measured in their local rest frames.
The resulting correction enters at order $\varpi^2$, the same order as
the displacement~\cite{beig2005rotation}. It must be retained even though
the angular velocity is small.
On the reference annulus $\mathcal D=\{X:A<|X|<B\}$, let
$x=R_z(\varpi t)\phi(X)$, where $R_z$ is a rotation about the $z$ axis.
For this calculation, put $F=D\phi$, $v=\varpi e_z\times\phi$, and
$X_\perp=(X^1,X^2,0)$. The rest-space strain and energy per reference
volume are
\begin{equation}
 C=F^{\mathsf T}\left(I+\frac{vv^{\mathsf T}}{1-|v|^2}\right)F,
 \qquad \epsilon(C)=\sqrt{\det C}\,\rho(n_i),
 \qquad n_i=\lambda_i(C)^{-1/2}.
 \label{eq:rotating-rest-strain}
\end{equation}
The extra term is the Lorentz correction along the local direction of
motion: the matrix in parentheses is the inverse of $I-vv^{\mathsf T}$,
the spatial part of the inverse material strain. The stationary material
action per unit time is $-\int_{\mathcal D}\sqrt{1-|v|^2}\epsilon(C)\,dX$.

\begin{proposition}[Leading centrifugal response]
\label{prop:centrifugal}
At $G=0$, write the centered, axially aligned branch as
$\phi(X)=X+\varpi^2u(X)+O(\varpi^4)$. Put
$\lambda=\kappa-2\mu_0/3$, $w=e_z\times X$ and
$\tau(E)=\lambda\operatorname{tr}(E)I+2\mu_0E$.
The coefficient $u$, modulo rigid motions, is the unique solution of
\begin{align}
 \mu_0\Delta u+(\lambda+\mu_0)\nabla\operatorname{div}u
   &=-(m_0+\lambda-\mu_0)X_\perp,\nonumber\\
 \tau(e(u))n&=-\tfrac\lambda2|X_\perp|^2n
       \quad (|X|=A,B),
 \qquad e(u)=\tfrac12(Du+Du^{\mathsf T}).
 \label{eq:centrifugal-boundary-problem}
\end{align}
Here $\tau$ is positive in tension; the physical principal pressures have
the opposite sign. The source excites only the spherical monopole and
spheroidal quadrupole.
\end{proposition}
\begin{proof}
Equation~\eqref{eq:rotating-rest-strain} gives
$C=I+2\varpi^2E+O(\varpi^4)$, where
$E=e(u)+w\otimes w/2$.
Since $D\epsilon(I)=0$, the $u$-dependent coefficient of $\varpi^4$ in the
reduced energy is
\begin{equation}
 \int_{\mathcal D}\left[
 \frac\lambda2(\operatorname{tr}E)^2+\mu_0|E|^2
                    -m_0X_\perp\cdot u\right]dX.
 \label{eq:centrifugal-functional}
\end{equation}
Variation gives $\operatorname{div}\tau(E)=-m_0X_\perp$ and
$\tau(E)n=0$. Direct differentiation yields
$\operatorname{div}\tau(w\otimes w/2)=(\lambda-\mu_0)X_\perp$;
on either sphere, $w\cdot n=0$. These identities give the displayed
equations. The loads have zero total force and torque by axial and
reflection symmetry. The positive elastic form and Korn's inequality~\cite{nitsche1981}
give existence and uniqueness on the rigid-motion quotient.
The angular decomposition of $X_\perp$ and $|X_\perp|^2n$ contains only
the stated two harmonics.
\end{proof}

Write $s=|X|/A$ and $P_2=(3\cos^2\theta-1)/2$. The displacement has the form
\begin{equation}
 u/A^3=(U_0+U_2P_2)e_r+V_2\partial_\theta P_2\,e_\theta.
 \label{eq:centrifugal-decomposition}
\end{equation}
The monopole $U_0$ changes the mean radius; the quadrupole $U_2$ changes the
polar radius relative to the equatorial radius, and
$V_2$ describes tangential material motion. Appendix~\ref{app:centrifugal}
reduces the boundary problem to six linear equations and gives their exact
solution for the selected moduli and $B=2A$.
A useful measure of each face's deformation is
\begin{align}
 \frac{r_{\rm eq}-r_{\rm pole}}{As}&
 =\varepsilon_\Omega\,\mathcal F(s)+O(\varepsilon_\Omega^2),\qquad
 \mathcal F(s)=-\frac{3U_2(s)}{2s},\nonumber\\
 \mathcal F(1)&\simeq17.56,\quad \mathcal F(2)\simeq8.54.
 \label{eq:face-flattening}
\end{align}
Both faces are oblate, with the larger fractional flattening at the inner
face. The surfaces are not scaled copies of one another. These response
coefficients depend on thickness and elastic moduli; they are not universal
numbers. Their strict signs persist in a neighborhood of the selected
parameters by continuous dependence of the invertible boundary problem.
At small positive $G$, the shape increment has corrections
$O(G\varpi^2+\varpi^4)$ in fixed units.

\subsection{The tide measured in the cavity}

At first order in $G$, write $g_{00}=-(1+2\Phi_{\rm lin})$.
The leading quadrupolar part of the potential in the regular cavity is
$\Phi_2=q r^2P_2(\cos\theta)$. Its measurable content is the electric tidal
tensor, which determines relative free-fall acceleration. In a locally
inertial orthonormal frame at the center,
$R_{\hat0\hat i\hat0\hat j}=\partial_i\partial_j\Phi_2=
q\operatorname{diag}(-1,-1,2)$, up to
$O(G^2\varpi^2+G\varpi^4)$ on compact subsets of the cavity.
Let $\xi^{\hat i}$ be an initially comoving tracer's separation from the
central geodesic
and $\tau$ the latter's proper time. Geodesic deviation gives, to this order,
\begin{equation}
 \frac{d^2}{d\tau^2}
 \begin{pmatrix}\xi^{\hat x}\\\xi^{\hat y}\\\xi^{\hat z}\end{pmatrix}
 =q\begin{pmatrix}\xi^{\hat x}\\\xi^{\hat y}\\-2\xi^{\hat z}\end{pmatrix}.
 \label{eq:cavity-tracer-deviation}
\end{equation}
Thus the sign of $q$ distinguishes stretching along the rotation axis
from stretching in the equatorial plane.
The corresponding error in the dimensionless tensor
$A^2R_{\hat0\hat i\hat0\hat j}$ is
$O(\chi^2\varepsilon_\Omega+\chi\varepsilon_\Omega^2)$ at fixed
$(\beta,k,\mu)$. No numerical bound on its coefficient is assumed.

The stationary Einstein equation for this potential is
$\Delta\Phi_{\rm lin}=4\pi G(T_{00}+T_{ii})$.
Indeed, stationarity gives $R_{00}=\Delta\Phi_{\rm lin}$ at linear
order, and the trace-reversed Einstein equation gives
$R_{00}=4\pi G(T_{00}+T_{ii})$ at the same order.
Here $T_{ii}$ is summed over the three spatial directions: pressure as well
as energy density sources this potential. Pulling the source back to
the reference annulus gives
\begin{equation}
 (T_{00}+T_{ii})(\phi)\det F
 =m_0+\varpi^2\left[-3\kappa\operatorname{div}u
             +\tfrac32(m_0-\kappa)|X_\perp|^2\right]+O(\varpi^4).
 \label{eq:rotating-active-density}
\end{equation}
This includes rest density, stress, kinetic energy and the moving volume.
It follows from $\rho=m_0[1-\varpi^2\operatorname{tr}E]+O(\varpi^4)$,
$\sum p_i=-3\kappa\varpi^2\operatorname{tr}E+O(\varpi^4)$ and the
two boost contributions $m_0\varpi^2|w|^2$ to $T_{00}$ and $T_{ii}$.

To determine $q$ from the material, expand the gravitational Green function.
For $r<|x'|$,
$|x-x'|^{-1}=\sum_{\ell\ge0}r^\ell|x'|^{-\ell-1}P_\ell(\cos\gamma)$,
where $\gamma$ is the angle between $x$ and $x'$. Axisymmetry therefore gives
\begin{equation}
 q=-G\int_{\rm shell}(T_{00}+T_{ii})(x')
                   \frac{P_2(\cos\theta')}{|x'|^3}\,d^3x'
 \label{eq:cavity-inverse-source}
\end{equation}
at first order in $G$. This is an inverse radial moment of the active
source. At the same weak-field order, the exterior mass quadrupole weights the
source by $|x'|^2$;
the cavity is more sensitive to matter near the inner face.
The spatially constant part of the cavity potential changes the central
clock rate and drops out of the tidal tensor. We first isolate the
quadrupole, then calculate that independent clock correction in
Section~\ref{sec:rotating-clock}.

Pull the integral back to the material annulus and expand its kernel as
well as its source. Source transport contributes
$m_0u\cdot\nabla(R^{-3}P_2)$. The angular identities
$\int P_2^2d\Omega=4\pi/5$ and
$\int|\nabla_{S^2}P_2|^2d\Omega=24\pi/5$ give
\begin{equation}
 \begin{gathered}
 \mathcal T=\int_1^\beta\frac{-3U_2+6V_2}{s^2}\,ds,
 \qquad\mathcal C=\int_1^\beta\frac{D_2}{s}\,ds,
 \qquad D_2=U_2'+2U_2/s-6V_2/s,\\
 q=-\frac{4\pi}{5}Gm_0A^2\varpi^2
 \left[\mathcal T-3k\mathcal C
              -(1-k)\frac{\beta^2-1}{2}\right].
 \end{gathered}
 \label{eq:cavity-source-parts}
\end{equation}
The first term accounts for the redistribution of relaxed mass by the
displacement. The second comes from the
pressure trace in Eq.~\eqref{eq:rotating-active-density}. The last includes
both motion and the pressure induced by the relativistic rest-space strain:
$|X_\perp|^2=2R^2(1-P_2)/3$ supplies its quadrupolar part.
Here $D_2$ is the quadrupolar coefficient of the displacement divergence;
positive $D_2$ denotes local volume expansion in that harmonic.

The definition of $D_2$ also gives
$(-3U_2+6V_2)/s^2=(U_2/s)'-D_2/s$. Integration and
Eq.~\eqref{eq:face-flattening} now express the transport in terms of
the faces:
\begin{equation}
 \mathcal T=\left[\frac{U_2}{s}\right]_1^\beta-\mathcal C
 =\frac23\bigl[\mathcal F(1)-\mathcal F(\beta)\bigr]-\mathcal C.
 \label{eq:cavity-face-moment}
\end{equation}
Combining the terms gives the quadrupole at order $G\varpi^2$ for any fixed
annulus and positive causal moduli in the small-parameter branch:
\begin{align}
 q&=-Gm_0A^2\varpi^2\mathcal Q_c,\nonumber\\
 \mathcal Q_c&=\frac{4\pi}{5}\left\{
 \frac23\bigl[\mathcal F(1)-\mathcal F(\beta)\bigr]
 -(1+3k)\mathcal C-(1-k)\frac{\beta^2-1}{2}\right\}.
 \label{eq:cavity-quadrupole}
\end{align}
The face term measures differential fractional flattening; the bulk term
combines volume transport with its induced pressure. Shear stiffness enters
through $U_2,V_2$. The response is fixed by the annulus and its elastic moduli.
This decomposition uses the reference material coordinates of the
weak-field construction. Its separate terms describe contributions to one
source integral; the observable is the resulting tidal tensor.

If $D_2=0$ and the faces have equal fractional flattening, the transport
term vanishes even when both faces are oblate. This cancellation concerns
transport only; the remaining active stresses still contribute to the tide.
Oblateness alone does not determine its sign.

\subsection{Two materials, opposite tides}

At $(\beta,k,\mu)=(2,1/5,1/10)$, the three terms in braces in
Eq.~\eqref{eq:cavity-quadrupole} are approximately
$6.011$, $-1.322$, and $-1.200$. Differential flattening dominates, so
$\mathcal Q_c\simeq8.77>0$ and $q<0$ at leading order.
Appendix~\ref{app:centrifugal} evaluates the integrals exactly.
The coefficient depends on thickness and moduli; its strict sign persists
locally by continuity of the invertible elastic boundary problem.

The sign can reverse within the same family of elastic laws while
maintaining causality near relaxation.
Keep $\beta=2$ and choose $k=\mu=2/5$. The relaxed squared sound speeds
are $2/5$ and $14/15$, so the small-coupling existence argument still
applies. The two faces remain oblate, with
$\mathcal F(1)\simeq3.06$ and $\mathcal F(2)\simeq1.45$, but now
$\mathcal Q_c\simeq-0.866$ and $q>0$. The exact coefficients are given in
Appendix~\ref{app:centrifugal}. Differential flattening is too small to
overcome the bulk and motion terms. A linear interpolation between
these two choices stays within the causal relaxed-modulus domain and
therefore passes through a zero of the leading quadrupole coefficient.
That cancellation leaves higher-order curvature undetermined; it does not
establish an exactly flat rotating cavity.

For either explicit material, smooth dependence and the nonzero leading
tidal coefficient rule out a flat cavity in a sufficiently small
two-parameter neighborhood with $G>0$ and $\varpi\ne0$.
The uniform remainder follows by factoring the central electric curvature
by $G\varpi^2$: it vanishes at $G=0$ and at $\varpi=0$, is even in
$\varpi$, and depends smoothly on both parameters in the regular cavity.
The remaining coefficient is continuous and has the nonzero limit computed
above. This establishes a neighborhood, rather than a numerical range of
angular velocities.
The first material stretches tracer separations along the rotation axis
and focuses them in the equatorial plane; the stiffer material reverses
these directions. This curvature cannot be removed by a change of
coordinates.

The zero trace of the tidal tensor expresses $R_{\hat0\hat0}=0$ in vacuum.
For either sign of nonzero $q$, at least one direction stretches rather
than restores the separation. The leading tide therefore supplies no
three-dimensional passive confinement near the center. This is a statement
about freely falling test particles; it does not decide whether the shell's
elastic modes grow. At $q=0$, higher orders must be examined.

Figure~\ref{fig:rotation}(b--d) shows both oblate faces, these relative
accelerations, and the source contributions. Its exaggerated display amplitude illustrates the leading coefficients;
it supplies no finite-rotation bound on strain or curvature.

\begin{figure}[tbp]
 \centering
 \includegraphics[width=\linewidth]{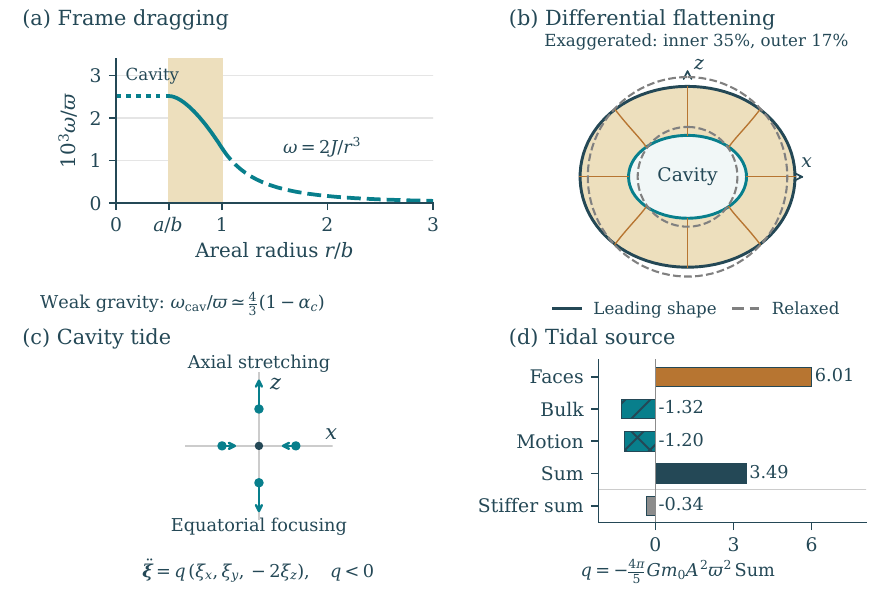}
 \caption{Rotational response of the shell in Eq.~\eqref{eq:elastic-example}.
 (a) Frame dragging $\omega$ at first order in the material angular
 velocity $\varpi$, matched across both vacuum faces and measured relative
 to infinity.
 (b) Leading centrifugal displacement at $G=0$, $B=2A$, drawn at equal
 scales with exaggerated amplitude $\varepsilon_\Omega=(A\varpi)^2=0.02$.
 The fractional flattenings $\varepsilon_\Omega\mathcal F$ are 35\% (inner)
 and 17\% (outer), relative to the relaxed radii. Solid faces include
 radial and tangential displacement; dashed circles are relaxed faces,
 and interior curves follow material meridians.
 (c) Local relative free-fall acceleration at order $G\varpi^2$.
 Dots mark equal initial separations from the central geodesic;
 arrows share an arbitrary scale.
 (d) Dimensionless source terms in Eq.~\eqref{eq:cavity-quadrupole}:
 Faces denotes differential fractional flattening; Bulk includes volume
 transport and pressure; Motion includes kinetic energy and the
 relativistic rest-space correction. Only the sum determines $q$.
 The stiffer material, $k=\mu=2/5$, gives $q>0$ at the same thickness
 despite two oblate faces. Panels (b--d) show leading coefficients
 without a finite-rotation error bound.}
 \label{fig:rotation}
\end{figure}

\subsection{The rotation-induced central clock shift}
\label{sec:rotating-clock}

Let $\alpha_{\rm cen}(G,\varpi)=\sqrt{-g(K,K)}$ at the center of the
reflection-symmetric rotating shell, with $K$ normalized at infinity.
The center follows a geodesic by symmetry. Its proper clock rate is
$d\tau/dt=\alpha_{\rm cen}$, and $\alpha_{\rm cen}(G,0)=\alpha_c$.
At leading order in gravity the Green function at the center gives
\begin{equation}
 \Phi_{\rm lin}(0)=-G\int_{\rm shell}
                 \frac{T_{00}+T_{ii}}{|x'|}\,d^3x'.
 \label{eq:central-clock-source}
\end{equation}
Pull back to material coordinates, as for the tide. The kernel changes by
$\delta(1/|\phi|)=-\varpi^2u_r/R^2$.
Only the spherical averages contribute: $\langle u_r\rangle=A^3U_0$,
$\langle\operatorname{div}u\rangle=A^2D_0$, where
$D_0=U_0'+2U_0/s$, and $\langle\sin^2\theta\rangle=2/3$.
Using Eq.~\eqref{eq:rotating-active-density}, we obtain
\begin{align}
 \alpha_{\rm cen}(G,\varpi)-\alpha_c
 &=\chi\varepsilon_\Omega\mathcal Z_c
   +O(\chi^2\varepsilon_\Omega+\chi\varepsilon_\Omega^2),\nonumber\\
 \mathcal Z_c&=4\pi\int_1^\beta
       \left[U_0+3k sD_0-(1-k)s^3\right]ds.
 \label{eq:rotating-clock-response}
\end{align}
The terms arise from displacement of the mass, the pressure trace, and
motion including the relativistic rest-space strain. Here a positive
$\mathcal Z_c$ means that rotation makes the central clock run faster
relative to infinity than in the static shell with the same material
content. The asymptotic error is at fixed $(\beta,k,\mu)$.

The original material has $\mathcal Z_c=16472\pi/2625\simeq19.71$.
The stiffer choice $k=\mu=2/5$ instead gives
$\mathcal Z_c=-15929\pi/7350\simeq-6.81$.
Appendix~\ref{app:centrifugal} gives the elementary integral evaluation.
Both materials gain rotational mass, yet their clock shifts have opposite
signs. There is no conflict: the mass uses the unweighted active source,
the central potential weights it by $1/|x'|$, and the tide weights its
quadrupole by $P_2/|x'|^3$. Moving material outward can weaken the central
potential even while the total energy increases. These leading coefficients
do not give a finite-rotation error bound or a universal relation between
the signs of the clock shift and the tide.

\section{Dynamics and physical scope}
\label{sec:stability-summary}

An equilibrium and its response to steady rotation do not determine how
the shell evolves after a disturbance. We now distinguish linear mode
stability from smooth nonlinear evolution, then consider what uniform
translation changes in the cavity measurements.

\subsection{Linear perturbations}

Radial perturbations
expand or contract spherical layers. Odd-parity
perturbations twist them, while nonspherical even-parity perturbations
change their shape and density. These sectors decouple on the static
spherical background, so stability in one does not decide the others.
Stability requires an energy estimate for the coupled matter and
gravitational field, including both moving vacuum faces. The interface
conditions are
\begin{equation}
 \Delta(T_{\mu\nu}n^\nu)=0,\qquad
 [\Delta h_{ab}]=0,\qquad [\Delta\mathcal K_{ab}]=0,
 \label{eq:general-linear-matching}
\end{equation}
where $h_{ab}$ and $\mathcal K_{ab}$ are the induced metric and extrinsic
curvature, brackets use the same normal on both sides, and $\Delta$ includes
the surface displacement and normal variation. In spherical motion the
material condition is $\Delta p_r=0$, rather than a condition at a fixed
coordinate radius.

For radial perturbations, eliminating the metric constraints gives a
self-adjoint equation for the material displacement, with gravitational
backreaction included. A positive energy bound then gives a lower bound
on squared oscillation frequencies. The sufficient criterion in
Eq.~\eqref{eq:radial-gap} holds at sufficiently small coupling;
its sampled value for
Eq.~\eqref{eq:elastic-example} is about $1.8\times10^{-3}L_0^{-2}$.
For odd parity, positive tangential enthalpy and shear stiffness give a
nonnegative coupled energy for every angular harmonic, including exterior
gravitational waves. This excludes exponentially growing physical modes
in that sector, but gives no positive lower bound on radiative frequencies
or decay rate. These results are derived in
Appendices~\ref{sec:radial-stability} and \ref{sec:axial-energy}.

Nonspherical even-parity and rotating stability remain unresolved.
For perfect-fluid thin shells with a flat interior, Ref.~\cite{pitre2026}
finds growing even-parity shape modes at every sampled compactness,
adiabatic index, and multipole $\ell\ge2$. Elastic membranes
on a prescribed Schwarzschild background can also be radially stable but
unstable to dipolar displacements~\cite{mourao2025}. Neither result decides
the stability of this finite-thickness, self-gravitating solid. Resolving
its even-parity sector requires the coupled metric and elastic perturbations,
with both faces displaced and matched to the two vacuum regions.
The limit of vanishing thickness at fixed surface mass is singular:
it changes the density scale, elastic response and interface equations.
Neither these thin-shell instabilities nor the present stability estimates
can be transferred across that limit without a separate derivation.

\subsection{Smooth evolution and wave breaking}

Conservation of ADM mass bounds a nonlinear radial energy norm while
strain and velocity remain near equilibrium. Control over a prescribed
finite time also requires small spatial derivatives and initial data
compatible with the gravitational constraints and free faces. Under these
conditions, the higher-regularity estimate controls the solution through
reflections. The permitted size of the initial data depends on the chosen
time interval.

Small amplitudes alone do not ensure this control. A sufficiently narrow
pulse can have small strain and velocity but large gradients.
Theorem~\ref{thm:shell-steepening} constructs such radial pulses whose first
derivatives diverge before either face is reached. Their steepening dominates
the bounded gravitational sources while the state remains causal and obeys
strict DEC. This is consistent with the finite-time estimate, whose $H^s$
norm controls several spatial derivatives. The theorem concerns the first
loss of smoothness, not an exponentially growing linear mode.
Discontinuous continuation and long-time smooth
evolution remain open.

\subsection{Uniform motion and additional loads}

Uniform translation gives an exact comparison beyond these perturbative
calculations. Apply an asymptotic Lorentz boost to both the metric and
material map. This gives a
coasting shell with $V^i=P^i/P^0$ and the same comoving cavity measurements.
Appendix~\ref{sec:coasting} gives the construction and its asymptotic
hypotheses. The boost constructs uniform motion, not an acceleration
process. Changing the shell's momentum requires exchange with another
system or with outgoing radiation; the total ADM momentum of the complete
isolated system remains conserved. Radiative steering~\cite{le2026recoil}
introduces photon emission and a different material law.

A massive payload changes the vacuum-cavity problem; supports holding it
relative to the shell also introduce tractions. An external tidal field
changes the gravitational boundary data even when the cavity stays empty.
The response must be found from the coupled metric and elastic equations,
as in relativistic tidal calculations for elastic
crusts~\cite{gittins2020}. Love numbers characterize the induced exterior
multipoles in response to an applied tidal field; the rotation-induced
cavity tide calculated here answers a different question. The locally flat cavity of
the isolated static solution does not establish gravitational shielding.

\section{Conclusions}

The leading cavity tide and rotational clock shift measure distinct source
moments of an elastic shell. The tide depends on differential flattening
of the two faces, deformation within the shell, and the stress and energy
of rotation. The clock shift depends on an inverse-radius moment of the
active source. For two choices of elastic moduli, each causal near
relaxation, both observables reverse sign while both faces remain oblate
and the total mass increases.
The cavity geometry therefore resolves differences in material response
that qualitative shape and total mass do not determine.

These results follow from the same constitutive action that supplies the
shell's static support: tangential compression balances gravity while
radial traction vanishes at both free faces. The static cavity is locally
flat, with clocks redshifted relative to infinity. In weak gravity, this
redshift and linear frame dragging measure the same inverse-radius mass
moment: the dragging fraction $\omega_{\rm cav}/\varpi$ is $4/3$ times the
clock-rate deficit $1-\alpha_c$ at leading order. Tidal curvature first
appears at second order in rotation.
The leading vacuum tide has stretching and focusing directions and
provides no three-dimensional passive confinement near the center.

The equilibrium and response calculations hold near relaxation. A complete
linear stability analysis still requires the nonspherical even-parity
sector with both moving faces and gravitational backreaction. Predictions
at a specified angular velocity require a bound on the rotational
remainder and on the resulting strain; the leading clock and tidal
coefficients do not supply those bounds. A payload or an external tide
changes the boundary-value problem. Energy conservation does not prevent
wave breaking; evolution beyond that point requires an
admissible weak-solution theory or additional material physics.
The construction describes the support, geometry, and partial stability
of an ideal gravitating solid. A material realization and a process that
forms or accelerates the shell require additional physics.

The sign change raises a concrete question: can the leading cavity tide
vanish on a rotating branch stable to nonspherical perturbations?
At a zero of $\mathcal Q_c$, the higher-order cavity curvature must also
be determined. Together with bounds on strain and the rotational
remainder, this would establish how nearly flat the cavity can remain
over a finite range of angular velocities.
A complementary inverse problem asks which combinations of bulk and shear
response can be inferred from joint clock and tidal measurements for a
known reference annulus and density. The observables probe distinct source
moments, but uniqueness of this inference and its sensitivity to the
assumed constitutive law remain to be determined.

\appendix
\section{An exact smooth anisotropic shell}
\label{sec:exact}

\begin{theorem}[Tangentially supported shell]
\label{thm:cluster}
Let $\rho\ge0$ be a $C^2$ function supported in $[a,b]$, with $0<a<b$,
and define $m(r)=4\pi\int_0^r\rho(s)s^2ds$, $M=m(b)$.
Suppose $2m(r)<r$ for $r>0$. Set
\begin{equation}
 p_r=0,\qquad p_t=\frac{\rho m}{2(r-2m)},\qquad
 \Phi(r)=-\int_r^\infty\frac{m(s)}{s[s-2m(s)]}\,ds.
 \label{eq:cluster}
\end{equation}
Then Eq.~\eqref{eq:spherical} solves all Einstein equations. The cavity is
flat and redshifted, the exterior is Schwarzschild, and neither interface
has a surface layer. Where $\rho>0$, DEC holds if and only if $5m\le2r$.
NEC and WEC hold everywhere.
\end{theorem}
\begin{proof}
The first two equations in Eq.~\eqref{eq:spherical-einstein} follow directly.
With $p_r=0$, Eq.~\eqref{eq:anisotropic-tov} becomes
$2p_t/r=\rho\Phi'$, which is exactly Eq.~\eqref{eq:cluster}.
Thus the angular equation holds as well. In the cavity $m=0$ and $\Phi$ is
constant. Outside $b$, $m=M$ and $e^{2\Phi}=1-2M/r$.
The prescribed regularity and $p_r=0$ give smooth matching. Both $\rho$ and
$p_t$ are nonnegative, and $p_t\le\rho$ reduces to $5m\le2r$ where
$\rho>0$. At zero density all stresses vanish.
\end{proof}

A compact example is obtained directly from the mass profile
\begin{equation}
 m(r)=M H\!\left(\frac{r-a}{b-a}\right),\qquad
 H(s)=35s^4-84s^5+70s^6-20s^7\quad(0<s<1),
 \label{eq:profile}
\end{equation}
extended by $H=0$ to the left and $H=1$ to the right.
Since $H'=140s^3(1-s)^3$, the density is nonnegative and $C^2$ at both
interfaces. The sufficient global bound $5M<2a$ ensures DEC.
A $C^\infty$ example follows by replacing $H'$ by a normalized positive
bump $\exp[-1/(s(1-s))]$ on $(0,1)$; the same compactness bound applies.
For example, $(a,b,M)=(10,20,1)$ satisfies this bound. This exact
construction establishes that a vacuum interface need not violate DEC.

An Einstein-cluster interpretation is available when the supporting particles
have timelike circular orbits. Their local squared speed is
\begin{equation}
 v_{\mathrm{circ}}^2=r\Phi'=\frac{m}{r-2m}<1
 \quad\Longleftrightarrow\quad 3m<r.
 \label{eq:orbits}
\end{equation}
An isotropic distribution of orbital planes has
$p_t=\rho v_{\mathrm{circ}}^2/2$, reproducing Eq.~\eqref{eq:cluster}.
This ideal distribution is concentrated on circular orbits; it is not a
smooth distribution on the full particle phase space. The stability results in Section~\ref{sec:stability-summary} concern
the elastic action~\eqref{eq:elastic-law}, which specifies the response
to perturbations.

\subsection{An inner pressure jump}

The obstruction in Proposition~\ref{prop:isotropic} cannot be removed by
truncating a positive-pressure fluid solution. Integrating Eq.~\eqref{eq:isotropic-tov}
inward from $p(b)=0$ generally produces $p(a)>0$. Replacing that value by
zero in the cavity creates a new matching problem. At a sharp inner interface,
$m(a)=0$ gives $\Phi'(a^+)=4\pi a p(a^+)$ while the vacuum cavity has
$\Phi'(a^-)=0$. With the normal pointing toward increasing $r$, the Israel surface tensor
has $\sigma=0$ and tangential pressure $P_\Sigma=a p(a^+)/2$. Thus a
nonzero positive inner pressure produces a massless surface layer that
violates surface DEC. Smoothing the jump
instead produces a transition stress that must be included explicitly; it
does not retain the original isotropic equilibrium.

\section{Explicit constitutive bounds}
\label{app:material-bounds}

The local sound speeds follow from two quadratic forms: the kinetic matrix
$W$ gives the inertia of a material disturbance, and $H(k)$ gives its
restoring stiffness for wave covector $k$. Their generalized eigenvalues
are the squared sound speeds. We derive both matrices and bound them
uniformly over the strain range used in the example.

Fix $\kappa/m_0=1/5$, $\mu_0/m_0=1/10$, and divide all densities,
stresses, and acoustic matrices in this appendix by $m_0$.
Write $(a,b,c)=(n_1,n_2,n_3)$ and $n=abc$. For a positive strain matrix $\mathsf B$,
\begin{equation}
 s^2=\frac{\operatorname{tr}\mathsf B\operatorname{tr}(\mathsf B^{-1})-9}{12},
 \qquad \widehat\rho=n+\frac{(n-1)^2}{10}+\frac{ns^2}{10}.
 \label{eq:invariant-material-law}
\end{equation}
Here $\mathsf B$ has eigenvalues $n_i^2$; the deformation tensor $C$ in
Eq.~\eqref{eq:rotating-rest-strain} has reciprocal eigenvalues. The invariant
is unchanged by inversion and is smooth also at repeated eigenvalues.
Its principal stress, with cyclic permutations for the others, is
\begin{equation}
 \widehat p_1=\frac{a^3b/c+a^3c/b+6a^2b^2c^2-6-b^3c/a-bc^3/a}{60}.
 \label{eq:explicit-principal-pressure}
\end{equation}
For the spherical equations, the same law reduces to
\begin{align}
 \widehat\rho(\nu,\eta)&=\frac{\nu^4+6\nu^3\eta^4+46\nu^2\eta^2+6\nu+\eta^4}{60\nu},\nonumber\\
 \widehat p_r(\nu,\eta)&=\frac{\nu^4+3\nu^3\eta^4-3\nu-\eta^4}{30\nu},\qquad
 \widehat p_t(\nu,\eta)=-\frac{\nu^4-6\nu^3\eta^4+6\nu-\eta^4}{60\nu}.
 \label{eq:explicit-spherical-law}
\end{align}

The spatial acoustic Hessian follows by substituting
$F=\operatorname{diag}(a,b,c)+\epsilon\xi k^{\mathsf T}$ in
Eq.~\eqref{eq:invariant-material-law}, with $\mathsf B=FF^{\mathsf T}$, and taking
$\partial_\epsilon^2\widehat\rho|_{\epsilon=0}=\xi^{\mathsf T}H(k)\xi$.
For the kinetic term use $\mathsf B=FF^{\mathsf T}-\dot X\dot X^{\mathsf T}$.
Then $W$ is diagonal, $W_{ii}=(\widehat\rho+\widehat p_i)/n_i^2$.
Define
\begin{align}
 L(a,b,c)&=\frac{3ab/c+3ac/b+12b^2c^2+b^3c/a^3+bc^3/a^3}{60},\nonumber\\
 T(a,b,c)&=\frac{ab/c+2ac/b+2bc/a+c^3/(ab)}{60},\nonumber\\
 B(a,b,c)&=\frac{a^2/c+b^2/c+12abc^2+2c-a^2c/b^2-b^2c/a^2}{60},\nonumber\\
 W_1(a,b,c)&=\frac{1}{120}\bigl[3ab/c+3ac/b+24b^2c^2+b^3/(ac)\nonumber\\
 &\hspace{4em}+90bc/a+c^3/(ab)-b^3c/a^3-bc^3/a^3\bigr].
 \label{eq:explicit-acoustic-coefficients}
\end{align}
The entries are $H_{11}=L(a,b,c)k_1^2+T(a,b,c)k_2^2+T(a,c,b)k_3^2$,
$H_{12}=B(a,b,c)k_1k_2$, and $W_{11}=W_1(a,b,c)$; simultaneous
permutations of densities and indices give the remaining entries.
In particular, $H_0=|k|^2I/10+7kk^{\mathsf T}/30$ and $W_0=I$ at relaxation.

Here are rational bounds valid throughout $l\le a,b,c\le u$, where
$l=99/100$ and $u=101/100$:
\begin{equation}
\begin{array}{c|cc}
 \text{quantity}&\text{lower bound}&\text{upper bound}\\ \hline
 L&19/60&7/20\\
 T&19/200&21/200\\
 B&11/50&37/150\\
 W_1&24/25&26/25\\
 \widehat p_1&-1/100&1/100
\end{array}
\label{eq:rational-material-table}
\end{equation}
The table follows by bounding each Laurent monomial separately.
For a positive monomial $a^ib^jc^k$, put $l$ in factors with positive
exponents and $u$ in those with negative exponents for its lower bound;
reverse the choices for its upper bound. Reverse the bounds for a negative
coefficient and add. For example,
\begin{equation}
 \frac{19}{60}\le
 \frac{6l^2/u+12l^4+2l^4/u^3}{60}
 \le L\le
 \frac{6u^2/l+12u^4+2u^4/l^3}{60}\le\frac7{20}.
 \label{eq:material-bound-example}
\end{equation}
The other rows follow by the same substitutions in the displayed formulas.
The resulting inequalities hold simultaneously for all states in the box.

For $|k|=1$, the table bounds each diagonal entry of $H-H_0$ in magnitude
by $1/60$ and each off-diagonal entry by $|k_i k_j|/75$.
Since $\sum_{j\ne i}|k_i k_j|\le1$, every absolute row sum is at most
$1/60+1/75=3/100$. The symmetric matrix therefore satisfies
\begin{equation}
 \frac7{100}I\preceq H\preceq\frac{109}{300}I,\qquad
 \frac{24}{25}I\preceq W\preceq\frac{26}{25}I.
 \label{eq:material-matrix-bounds}
\end{equation}
The generalized Rayleigh quotient lies between $7/104>1/16$ and
$109/288<2/5$, proving the characteristic bound in
Eq.~\eqref{eq:material-box-bounds} for every polarization and direction.
Also $\widehat\rho\ge n\ge(99/100)^3$ and
$|\widehat p_i|\le1/100$, so the DEC slack exceeds $24/25$.

\section{Stationary material motion and staticity}
\label{app:stationary}

Stationarity means that the metric and stress do not change along a time
symmetry. Individual particles can still move, as they do in a rotating
body~\cite{andersson2010}. For an elastic solid, the material metric
restricts this motion because it records which particles are neighbors.

Let $q_{ab}=g_{ab}+u_au_b$ be the material rest metric and
$b_{ab}=k_{AB}\partial_aX^A\partial_bX^B$ the material metric pulled back
to spacetime. On the three-dimensional rest space, define
$\mathsf B^a{}_b=q^{ac}b_{cb}$ and the spatial stress
$\mathsf S^a{}_b=q^{ac}T_{cd}q^d{}_b$.
The eigenvalues of $\mathsf B$ are $n_i^2$.
Isotropy gives a smooth matrix function $\mathsf S=\mathcal S(\mathsf B)$.

\begin{proposition}[Stationary motion of material labels]
\label{thm:material-symmetry}
Consider an Einstein--elastic solution with the
law~\eqref{eq:elastic-law} and $\kappa,\mu_0>0$.
Let the particle worldlines of a single body be bijectively labeled by a
bounded connected Euclidean domain $\mathcal D$ with smooth boundary.
Suppose a complete spacetime Killing field $K$ preserves the material
worldtube, $\rho+p_i>0$, and all strains lie in a sufficiently small fixed
rotation-invariant neighborhood of $\mathsf B=I$.
Then, in material coordinates centered at the centroid of $\mathcal D$,
\begin{equation}
 KX^A=\Omega^A{}_B X^B,\qquad \Omega_{AB}=-\Omega_{BA},
 \label{eq:material-killing}
\end{equation}
where $\Omega$ is constant and its rotations preserve $\mathcal D$.
If $\mathcal D$ has no continuous rotational symmetry, then $KX^A=0$.
\end{proposition}
\begin{proof}
First, the stress determines the strain on this constitutive branch.
Differentiating Eq.~\eqref{eq:elastic-law} at relaxation gives, for every
symmetric matrix $E$,
\begin{equation}
 D\mathcal S_I[E]=\mu_0 E+
 \left(\frac\kappa2-\frac{\mu_0}3\right)\operatorname{tr}(E)I.
 \label{eq:stress-inverse}
\end{equation}
For diagonal $E$, substitute $n_i=1+\epsilon E_{ii}/2$ to first order;
orthogonal diagonalization and isotropy give the formula for general $E$.
The eigenvalues of this linear map on traceless and isotropic strains are
$\mu_0$ and $3\kappa/2$. Thus
$\langle E,D\mathcal S_I[E]\rangle\ge c|E|^2$ for some $c>0$, using the
matrix inner product $\langle E,F\rangle=\operatorname{tr}(E^{\mathsf T}F)$.
Continuity preserves a positive lower bound in a sufficiently small ball
about $I$. Integrating the derivative along each straight segment in that
ball proves that $\mathcal S$ is injective there. The ball can be chosen
rotation invariant and contained in the positive strain matrices.

The Einstein equations imply $\mathcal L_KT=0$. Since $\rho+p_i>0$,
$-\rho$ is distinct from every spatial stress eigenvalue; hence the future
unit timelike eigenvector $u$ is unique and $\mathcal L_Ku=0$.
The rest metric and spatial stress are therefore stationary. Isotropy and
injectivity of $\mathcal S$ imply $\mathcal L_K\mathsf B=0$, and consequently
$\mathcal L_Kb=0$.

Because $[K,u]=0$, the flow of $K$ maps entire particle worldlines to
particle worldlines. It descends through their labels to a flow on
$\mathcal D$. The equality $\mathcal L_K(X^*k)=0$ makes this a Euclidean
isometry. Its generator $\xi$ obeys
$\partial_A\xi_B+\partial_B\xi_A=0$. Differentiating and cycling the
indices gives $\partial_A\partial_B\xi_C=0$; connectedness then gives
$\xi=c+\Omega X$ with constant antisymmetric $\Omega$.
These isometries preserve the bounded domain and its volume, so they fix
its centroid. In centered coordinates $c=0$, proving the claim.
\end{proof}

For a round material annulus, the allowed label motion is a rotation.
In stationary coordinates $K=\partial_t$, Eq.~\eqref{eq:material-killing}
integrates to $X(t,x)=e^{t\Omega}f(x)$. Where $f$ is invertible, particle
paths are images under $f^{-1}$ of material circles. Their spatial shapes
and speeds need not be those of rigid rotation. The proposition therefore
imposes no axial spacetime symmetry and supplies no additional circulating
equilibrium.

The sector $KX=0$ has a stronger geometric restriction. Here the particles
follow the stationary time symmetry itself.

\begin{proposition}[Staticity of a comoving isolated shell]
\label{thm:staticity}
Suppose an isolated shell governed by the Einstein--elastic equations admits
the decomposition
\eqref{eq:stationary-quotient} globally on $\mathbb R\times\mathbb R^3$,
with time-independent fields, $K=\partial_t$, $e^{2U}>0$, and
positive-definite $h$.
Assume a regular vacuum cavity, piecewise smooth fields with Darmois
matching and no surface layers, and
$U,h-\delta,\psi=O(r^{-1})$ with first derivatives $O(r^{-2})$ at infinity.
If $KX^A=0$ throughout the material, then $d\psi=0$ and the spacetime is
static.
\end{proposition}
\begin{proof}
Put $F_{ij}=\partial_i\psi_j-\partial_j\psi_i$, with its indices raised
using $h$. The inverse metric and volume element are
\begin{equation}
 g^{-1}=-e^{-2U}\partial_t^2+
 e^{2U}h^{ij}(\partial_i-\psi_i\partial_t)
                  (\partial_j-\psi_j\partial_t),\qquad
 \sqrt{-g}=e^{-2U}\sqrt h.
 \label{eq:stationary-inverse-metric}
\end{equation}
When $\partial_tX=0$, both the strain
$H^{AB}=e^{2U}h^{ij}\partial_iX^A\partial_jX^B$ and the material volume
term are independent of $\psi$. The curvature identity for
Eq.~\eqref{eq:stationary-quotient} is~\cite{andersson2010}
\begin{equation}
 R_g\sqrt{-g}=\sqrt h\left(R_h+2\Delta_hU-2|DU|_h^2
                       +\tfrac14e^{4U}F_{ij}F^{ij}\right).
 \label{eq:stationary-curvature}
\end{equation}
Varying $\psi$ in the Einstein--material action therefore gives the
source-free mixed equation
\begin{equation}
 D_i(e^{4U}F^{ij})=0.
 \label{eq:stationary-sourcefree}
\end{equation}
Multiply by $\psi_j$ and integrate over the full spatial domain, including
the cavity. Antisymmetry gives
$(D_i\psi_j)F^{ij}=F_{ij}F^{ij}/2$, hence
\begin{equation}
 \frac12\int_{\mathbb R^3}e^{4U}F_{ij}F^{ij}\,dV_h
 =\lim_{R\to\infty}\int_{S_R}\psi_j e^{4U}F^{ij}n_i\,dA_h=0.
 \label{eq:staticity-energy}
\end{equation}
The surface term is $O(R^{-1})$. At each material face the two fluxes
cancel by the weak Einstein equation and the absence of a surface layer.
The regular center contributes no boundary term. Positivity of $h$ and
$e^{4U}$ implies $F=0$. Since the spatial domain is $\mathbb R^3$,
$\psi=df$ globally, and $T=t+f(x)$ removes the cross term in
Eq.~\eqref{eq:stationary-quotient}.
\end{proof}

This result allows arbitrary shell shape and anisotropic stress. It uses
strict stationarity, a regular center, the global product topology, and
asymptotically vanishing connection. It does not require stress injectivity.
The rotating branch evades its comoving hypothesis: the particles follow
$K+\varpi L$, not $K$. That helical generator becomes spacelike far from
the axis, so it cannot replace $K$ in the global positive metric
decomposition above.

\section{The small rotating branch}
\label{app:rotating-existence}

\begin{proof}[Proof of Theorem~\ref{thm:rotating}]
Theorem~4.6 of Ref.~\cite{andersson2010} applies to a bounded connected
axisymmetric material domain with smooth boundary. It does not require that
the domain contain the origin or that its boundary be connected. The annulus
therefore qualifies. The law depends only on the principal densities, so
it is isotropic and invariant under rotations of the material reference
frame, as required in the rotating construction. The relaxed specific energy
is $m_0>0$, and its elastic quadratic form is
$\kappa(\operatorname{tr}e)^2+2\mu_0|\operatorname{dev}e|^2>0$ for
nonzero symmetric $e$. These are the required constitutive conditions.
The construction's gauge anchor can be chosen on the rotation axis within
the annulus, for example at $X_0=(0,0,(A+B)/2)$.
The theorem solves the full equations and restores the projected force and
gauge equations. Its small deformation preserves both boundary components
and the vacuum cavity. DEC, the timelike material flow, and the characteristic
gaps persist from relaxation by continuity on the compact body.
The branch depends smoothly on $G$ and $\varpi$ in material coordinates.
Indeed, the fixed-domain residual in Section~3 of Ref.~\cite{andersson2010}
is smooth for this law: its coefficients use smooth strain invariants,
matrix inverses and determinants near positive matrices. A bounded extension
of the material displacement to all of $\mathbb R^3$, equal to zero outside
a compact set, fixes both interfaces when the equations are pulled back.
The unknowns lie in $W^{2,p}$ on the body and weighted $W^{2,p}_\delta$
spaces for the gravitational fields, with $p>3$ and $-1<\delta<-1/2$.
Their first derivatives are bounded, so the products of first derivatives
and coefficients multiplying second derivatives define smooth maps into
the corresponding $L^p$ spaces. The traction residual takes values in
$W^{1-1/p,p}$ on both faces. The material indicator is fixed under
parameter differentiation; one does not differentiate a moving density jump
in a fixed spatial chart. The projected derivative is an
isomorphism after removing rigid motions, so the smooth implicit-function
theorem applies; the subsequent equilibration restores the full equations.

The chosen point fixes translations for the existence proof. For the
cavity observables we translate along the axis to center the image of the
material annulus, fixing $\int_{\mathcal D}\phi(X)\,dX=0$, and fix its
orientation. These conditions remove the same rigid-motion kernel.
Equatorial reflection preserves the centered data, so local uniqueness
makes it a symmetry of the branch. Together with axial symmetry it fixes
the central worldline and sets its spatial acceleration to zero.
Time reversal changes $\varpi$ to $-\varpi$; scalar clock rates and the
central electric tidal eigenvalues are consequently even in $\varpi$.
On compact subsets of the cavity, the stationary vacuum equations are
elliptic in the chosen gauge. Interior regularity applied to their
parameter derivatives promotes the smooth dependence above to the
curvature and clock observables used in the main text.
\end{proof}

\section{Centrifugal response coefficients}
\label{app:centrifugal}

The angular decomposition reduces the displacement problem to radial
equations. The tractions on both faces fix their integration constants.
An integrated force balance first relates the stresses to the rotational
mass. Dot $\operatorname{div}\tau(E)=-m_0X_\perp$ with $X$
and integrate over the reference annulus. The free-traction terms vanish
at both faces, giving
\begin{equation}
 3\kappa\int_{\mathcal D}\operatorname{div}u\,dX
 =\left(m_0-\frac{3\kappa}{2}\right)
   \int_{\mathcal D}|X_\perp|^2\,dX.
 \label{eq:centrifugal-virial}
\end{equation}
Substitution in Eq.~\eqref{eq:rotating-active-density} gives the integrated
order-$\varpi^2$ term
\[
 \tfrac12m_0\varpi^2\int_{\mathcal D}|X_\perp|^2dX
 =\tfrac12\mathcal I_0\varpi^2,
\]
as required by the first law.
The spatial stresses cancel in the unweighted integral of $T_{ii}$ at
this order. They need not cancel in the inverse moment that determines
the cavity tide.

Use the dimensionless moduli $k,\mu$ of
Eq.~\eqref{eq:dimensionless-parameters}, and put
$\lambda_*=k-2\mu/3$, $f_*=1+\lambda_*-\mu$.
Substitution of Eq.~\eqref{eq:centrifugal-decomposition} into
Eq.~\eqref{eq:centrifugal-boundary-problem} gives
\begin{align}
 (\lambda_*+2\mu)(U_0''+2U_0'/s-2U_0/s^2)&=-2f_*s/3,\nonumber\\
 \mu(U_2''+2U_2'/s-8U_2/s^2+12V_2/s^2)
                  +(\lambda_*+\mu)D_2'&=2f_*s/3,\nonumber\\
 \mu(V_2''+2V_2'/s-6V_2/s^2+2U_2/s^2)
                  +(\lambda_*+\mu)D_2/s&=f_*s/3,
 \label{eq:centrifugal-radial-odes}
\end{align}
where $D_0=U_0'+2U_0/s$ and $D_2=U_2'+2U_2/s-6V_2/s$.
At both $s=1,\beta$, the three traction conditions are
\begin{equation}
 \lambda_*D_0+2\mu U_0'=-\lambda_*s^2/3,\qquad
 \lambda_*D_2+2\mu U_2'=\lambda_*s^2/3,\qquad
 V_2'+(U_2-V_2)/s=0.
 \label{eq:centrifugal-six-tractions}
\end{equation}
The radial equations have the general solution
\begin{align}
 U_0&=h_0s+j_0/s^2-\frac{f_*s^3}{15(\lambda_*+2\mu)},\nonumber\\
 U_2&=2a_2s+\frac{6\lambda_*b_2+f_*/3}{5\lambda_*+7\mu}s^3
       +\frac{(3\lambda_*+5\mu)c_2}{\mu s^2}-\frac{3d_2}{s^4},\nonumber\\
 V_2&=a_2s+b_2s^3+c_2/s^2+d_2/s^4.
 \label{eq:centrifugal-general-profiles}
\end{align}
The homogeneous radial powers are $s,s^{-2}$ in the monopole and
$s,s^3,s^{-2},s^{-4}$ in the quadrupole; a cubic particular solution
accounts for the source. Positivity of $k,\mu$
keeps the displayed denominators nonzero. The six traction conditions
fix the six integration constants: the homogeneous problem has zero
elastic energy and hence only a rigid displacement, absent in these
monopole and spheroidal quadrupole sectors.
For $(\beta,k,\mu)=(2,1/5,1/10)$ this gives
\begin{align}
 U_0&=\frac{5053}{3150}s+\frac{326}{175s^2}-\frac{31}{150}s^3,\nonumber\\
 U_2&=2a_2s+\left(\frac{24b_2}{41}+\frac{31}{123}\right)s^3
                       +\frac{9c_2}{s^2}-\frac{3d_2}{s^4},\nonumber\\
 V_2&=a_2s+b_2s^3+\frac{c_2}{s^2}+\frac{d_2}{s^4},\nonumber\\
 (a_2,b_2,c_2,d_2)&=
 \frac{(-11413363,\,2029314,\,-3398520,\,-4645536)}{3202122}.
 \label{eq:centrifugal-profiles}
\end{align}
At $s=1,2$, with $D_2=U_2'+2U_2/s-6V_2/s$, the quadrupole tractions reduce to
$2D_2/15+U_2'/5=2s^2/45$ and $V_2'+(U_2-V_2)/s=0$.
With $D_0=U_0'+2U_0/s$, the monopole traction is
$2D_0/15+U_0'/5=-2s^2/45$ at both faces.

In particular, $U_2(1)=-6246978/533687$ and
$U_2(2)=-6078274/533687$.
Direct integration of $D_2/s$ gives $\mathcal C=125979/152482$.
Thus the face and bulk terms of Eq.~\eqref{eq:cavity-quadrupole} are
$\tfrac23[\mathcal F(1)-\mathcal F(2)]=458263/76241$ and
$-(1+3k)\mathcal C=-503916/381205$, respectively.
The transport, pressure, and motion terms in
Eq.~\eqref{eq:cavity-source-parts} evaluate to
\begin{align}
 I_{\rm disp}&=\frac{790547}{152482},\qquad
 I_{\rm stress}=-\frac{377937}{762410},\qquad
 I_{\rm kin}=-\frac65,\nonumber\\
 \mathcal Q_c&=\frac{4\pi}{5}(I_{\rm disp}+I_{\rm stress}+I_{\rm kin})
             =\frac{5319812\pi}{1906025}>0.
 \label{eq:exact-cavity-response}
\end{align}
These fractions specify the solution of this particular boundary problem.
The physical conclusions use their signs and relative size. In particular,
$\mathcal Q_c$ is dimensionless; $q$ has dimensions of inverse length squared.

For the contrasting material $k=\mu=2/5$ and $\beta=2$, the same
six traction equations instead give
\begin{align}
 (h_0,j_0)&=\left(\frac{1891}{4410},\frac{61}{245}\right),\nonumber\\
 (a_2,b_2,c_2,d_2)&=\left(-\frac{1271651}{2185071},
 \frac{303495}{2913428},-\frac{193040}{728357},-\frac{150784}{728357}\right).
 \label{eq:stiffer-centrifugal-coefficients}
\end{align}
Substitute these constants in Eq.~\eqref{eq:centrifugal-general-profiles},
with $\lambda_*=2/15$ and $f_*=11/15$. The face flattenings and volume
moment are
\begin{equation}
 \mathcal F(1)=\frac{1272675}{416204},\qquad
 \mathcal F(2)=\frac{150630}{104051},\qquad
 \mathcal C=\frac{24495}{104051}.
 \label{eq:stiffer-face-response}
\end{equation}
Both flattenings are positive. The inverse moment, however, is
\begin{equation}
 \mathcal Q_c=\frac{4\pi}{5}\left(
 \frac{223385}{208102}-\frac{53889}{104051}-\frac9{10}\right)
 =-\frac{716848\pi}{2601275}<0.
 \label{eq:stiffer-cavity-response}
\end{equation}
Thus the sign reversal follows from the same boundary problem and source
moment as the original example. The uniform constitutive box in
Appendix~\ref{app:material-bounds} applies only to the original moduli;
for this second choice, causality and DEC follow in a sufficiently small
neighborhood of relaxation from their strict relaxed inequalities.

The monopole gives the central clock coefficient without any additional
field equation. Write $c_0=f_*/[15(\lambda_*+2\mu)]$ in
Eq.~\eqref{eq:centrifugal-general-profiles}; then
$U_0=h_0s+j_0/s^2-c_0s^3$ and $D_0=3h_0-5c_0s^2$.
Termwise integration in Eq.~\eqref{eq:rotating-clock-response} gives
\begin{equation}
 \frac{\mathcal Z_c}{4\pi}
 =\frac{(1+9k)h_0}{2}(\beta^2-1)
  +j_0(1-\beta^{-1})
  -\frac{(1+15k)c_0+1-k}{4}(\beta^4-1).
 \label{eq:central-clock-coefficient}
\end{equation}
For the first material, substitute Eq.~\eqref{eq:centrifugal-profiles}
and $c_0=31/150$ to obtain $\mathcal Z_c/(4\pi)=4118/2625$.
For the second material, $c_0=11/210$; inserting
Eq.~\eqref{eq:stiffer-centrifugal-coefficients} gives
$\mathcal Z_c/(4\pi)=-15929/29400$.

\section{Rotational first law}
\label{app:first-law}

Consider a regular, horizon-free member of the rigidly rotating branch.
Keep the material annulus, reference metric and constitutive law fixed.
The first law relates the change in mass to the change in angular momentum
between neighboring equilibria. Its derivation must include the moving
faces: although they do no traction work, their displacement changes the
domain occupied by the material.
For a variation $h_{\mu\nu}=\delta g_{\mu\nu}$, $h=g^{\mu\nu}h_{\mu\nu}$,
the Einstein--Hilbert and material actions have boundary potentials
\begin{equation}
 \theta_g^\mu=\frac{\sqrt{-g}}{16\pi G}
       (\nabla_\nu h^{\mu\nu}-\nabla^\mu h),\qquad
 \theta_m^\mu=-2\sqrt{-g}\,\rho_{H^{AB}}\nabla^\mu X^B\delta X^A.
 \label{eq:material-symplectic-potential}
\end{equation}
Here the indices on $h_{\mu\nu}$ are raised with the background metric.
At a moving face with displacement $\xi$, fixed labels give
$\delta X^A=-\xi^\nu\partial_\nu X^A$. The stress identity
\[
 2\rho_{H^{AB}}\nabla^\mu X^B\partial_\nu X^A
 =T^\mu{}_\nu+\rho\delta^\mu{}_{\nu}
\]
then yields
\begin{equation}
 n_\mu\theta_m^\mu-\sqrt{-g}\rho\,n\cdot\xi
   =\sqrt{-g}\,n_\mu T^\mu{}_\nu\xi^\nu=0.
 \label{eq:first-law-face}
\end{equation}
The second term is the variation of the material domain. Zero traction
removes its work at each face; Darmois matching cancels the gravitational
interface terms. The regular cavity supplies no inner boundary charge.

Antisymmetrizing the variation of the potentials gives the symplectic form
$\mathcal W_\Sigma$. Its sign is fixed by
$\mathcal W(\delta_1,\delta_2)=\int(\delta_1\pi\,\delta_2q-
\delta_2\pi\,\delta_1q)$ over the material and gravitational canonical pairs.
For variations satisfying the linearized constraints and matching, the
Hamiltonian identity~\cite{iyer1994} for a fixed generator
$k=K+\varpi L$ is
\begin{equation}
 \delta H_k=\mathcal W_\Sigma(\delta,\mathcal L_k(g,X)),
 \qquad H_k=\mathcal M-\varpi\mathcal J.
 \label{eq:rotational-hamiltonian}
\end{equation}
The charge follows from the ADM surface term at infinity; standard
asymptotically flat decay and parity conditions make it finite.
At a rigidly rotating equilibrium, $\mathcal L_k g=0$ and $kX=0$, so
$\delta\mathcal M-\varpi\delta\mathcal J=0$. Applying this identity at each
point of the equilibrium branch, with its generator fixed during that
variation, gives $d\mathcal M=\varpi\,d\mathcal J$ at fixed $G$. Differentiating
the smooth odd expansion of $\mathcal J$ and integrating this identity gives
Eq.~\eqref{eq:rotating-first-law}.

\section{Uniform translation}
\label{sec:coasting}

The velocity of a stationary translation is already restricted at infinity,
independently of the constitutive law or internal flow.

\begin{proposition}[Asymptotic translation~\cite{beig1996}]
\label{prop:asymptotic}
Let a spacetime have a Killing field $K=\ell n+Y$ on an asymptotically
Euclidean vacuum end of a spacelike hypersurface with future unit normal
$n$. For some $\alpha>1/2$,
assume
\begin{equation}
 \gamma_{ij}-\delta_{ij}=O_2(r^{-\alpha}),\quad
 K_{ij}=O_1(r^{-1-\alpha}),\quad
 \ell-a^0=O_1(r^{-\alpha}),\quad Y^i\to a^i,
 \label{eq:asymptotic-data}
\end{equation}
with $\ell,Y\in C^1$ and $a^0>0$. Here $O_k$ includes decay of derivatives
through order $k$. If the ADM four-momentum is future timelike, then
\begin{equation}
 P^\mu=\frac{P^0}{a^0}a^\mu,\qquad
 V^i:=\frac{a^i}{a^0}=\frac{P^i}{P^0},\qquad |\mathbf V|<1.
 \label{eq:velocity-momentum}
\end{equation}
\end{proposition}
\begin{proof}
The Killing equation gives $\mathcal L_Y\gamma_{ij}=2\ell K_{ij}$ in
our sign convention, and the vacuum constraints give $E=S_i=0$ on the end.
Proposition~3.1 of Ref.~\cite{beig1996}, with the opposite sign for
extrinsic curvature, therefore gives $P^\mu\parallel a^\mu$.
The remaining conclusions follow from $P^0,a^0>0$ and timelikeness of $P$.
\end{proof}
This established result fixes the asymptotic translation speed even for a
circulating shell. It neither forces staticity nor determines an internal
clock comparison. The explicit family below realizes it with comoving matter.

\begin{proposition}[Exact coasting family]
\label{prop:coasting}
Let $(\bar g,\bar X)$ be a spherical elastic equilibrium constructed above,
with geometrized ADM mass $M$. For every $|v|<1$, set
\begin{equation}
 T=\Gamma_v(t-vx),\quad X=\Gamma_v(x-vt),\quad Y=y,\quad Z=z,
 \qquad \Gamma_v=(1-v^2)^{-1/2}.
 \label{eq:boost-map}
\end{equation}
Pulling back both $\bar g$ and $\bar X$ gives an exact Einstein--elastic
shell stationary under $K=\Gamma_v(\partial_t+v\partial_x)$, with a flat
vacuum cavity, zero surface traction, and the same material energy-condition
and characteristic bounds. In the corresponding asymptotic inertial frame,
\begin{equation}
 P^\mu=(\Gamma_v M,\Gamma_v Mv,0,0),\qquad
 (P^0)^2-|\mathbf P|^2=M^2.
 \label{eq:coasting-momentum}
\end{equation}
\end{proposition}
\begin{proof}
General covariance preserves the full gravitational and material equations,
including the interface conditions, under the same pullback. The map sends
$K$ to $\partial_T$, so the fields are invariant along $K$ and the particles
follow it. Flatness, DEC, and local characteristic cones are invariant.
The asymptotic transformation is a Lorentz boost~\cite{beig1997}; applying it to the rest
ADM momentum $(M,0,0,0)$ gives Eq.~\eqref{eq:coasting-momentum}.
\end{proof}

The faces in the laboratory coordinates satisfy
$\Gamma_v^2(x-vt)^2+y^2+z^2=r(A)^2$ and $r(B)^2$, respectively.
Thus this is a uniformly moving material shell with nonzero momentum relative
to the chosen asymptotic frame. It is stationary with respect to $K$, not
with respect to laboratory $\partial_t$. Requiring both time symmetries at
nonzero $v$ would make the material support translation invariant along $x$,
which is incompatible with a nonempty compact shell. No small-velocity or
weak-field expansion is used in the boost.
Figure~\ref{fig:coasting} compares rest and laboratory worldtube sections. The straight
material paths in this coordinate diagram retain the static shell's proper
acceleration $\sqrt F\,\Phi'$; the central cavity tracer is geodesic.

\begin{figure}[t]
 \centering
 \includegraphics[width=\linewidth]{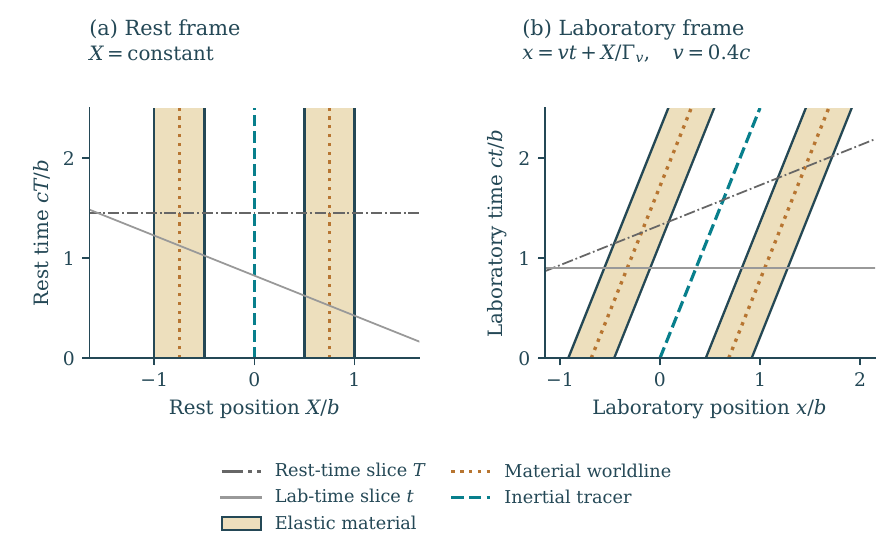}
 \caption{The same shell at rest (a) and coasting (b), in coordinate
 sections $Y=Z=0$ and $y=z=0$ with equal axis scales.
 The laboratory speed is $v/c=P^x/P^0=0.4$.
 Shading marks material; dotted paths are stress-supported material
 worldlines, and the dashed path is an inertial cavity tracer.
 The metric and material undergo the same exact Lorentz pullback,
 giving uniform motion with conserved ADM momentum. Lines of equal $T$
and equal $t$ show the different simultaneity slices. These are coordinate
sections, not measurements of proper shell thickness.}
 \label{fig:coasting}
\end{figure}

A test passenger initially comoving in the flat cavity remains inertial.
Pulling back the perturbations and their Cauchy surfaces carries the radial
and odd-parity energy estimates to the corresponding comoving slices.
The same pullback preserves the fixed-loop clock experiment of
Eq.~\eqref{eq:proper-sagnac}; a boost adds no asymmetry to that experiment.

\section{Gravitational constraints and stress conventions}
\label{sec:source}

With $G=1$, the ADM convention is
\begin{equation}
 ds^2=-\alpha^2dt^2+\gamma_{ij}(dx^i+\beta^i dt)(dx^j+\beta^j dt),
 \qquad n^a=\alpha^{-1}(\partial_t-\beta^i\partial_i)^a,
 \label{eq:adm}
\end{equation}
we use $K_{ij}=(-\partial_t\gamma_{ij}+D_i\beta_j+D_j\beta_i)/(2\alpha)$,
$E=T_{ab}n^an^b$, and $S_i=-\gamma_i{}^aT_{ab}n^b$. The constraints are
\begin{equation}
 {}^{(3)}R+K^2-K_{ij}K^{ij}=16\pi E,
 \qquad D_j(K^j{}_i-\delta^j_i K)=8\pi S_i.
 \label{eq:constraints}
\end{equation}
The spatial Einstein equations, material equations, and free-boundary
conditions complete this system.

The timelike material eigenvector makes the elastic stress Type I in every
observer frame. For a general source, causal eigenvectors and Jordan
structure must be identified at degeneracies~\cite{martinmoruno2018core}.

\section{Coupled radial dynamics}
\label{sec:radial-stability}

Radial motion changes the strain and gravitational field together. We
first eliminate the spherical metric constraints to obtain the oscillation
energy, then determine what its nonlinear counterpart controls.
Throughout this appendix $G$ is explicit.
For a displacement $\xi(t,r)$, the Eulerian variation $\delta$ compares
fields at fixed radius; the Lagrangian variation
$\Delta=\delta+\xi\partial_r$ follows the same material particle. Put
$w=\rho+p_r$, $q=p_t-p_r$, and define the material derivatives
\begin{equation}
 \mathcal A=n_r\partial_{n_r}p_r,\qquad
 \mathcal B=n_t\partial_{n_t}p_r,\qquad
 \mathcal D=n_t\partial_{n_t}p_t,
 \label{eq:radial-moduli}
\end{equation}
where $n_t$ varies both tangential densities together. The linearized Einstein
constraints, with no cavity mass perturbation, give
\begin{equation}
 \delta m=-4\pi G r^2w\xi,\qquad
 \Delta\Lambda=-\Phi'\xi,\qquad
 \Delta p_r=-\mathcal A(\xi'-\Phi'\xi)-\frac{\mathcal B}{r}\xi.
 \label{eq:radial-constraints}
\end{equation}
Both moving faces satisfy $\Delta p_r=0$. Holding their Eulerian pressure
perturbation at zero would give a different boundary-value problem.

Set $\zeta=r^2e^{-\Phi}\xi$ and
\begin{align}
 P&=\frac{\mathcal A e^{\Lambda+3\Phi}}{r^2},\qquad
 W=\frac{w e^{3\Lambda+\Phi}}{r^2},\qquad
 Y=\frac{2-\mathcal B/\mathcal A}{r},\nonumber\\
 Q&=-\frac{e^{\Lambda+3\Phi}}{r^2}\mathcal V,\qquad
 \mathcal V=\frac{2q+2\mathcal D-\mathcal B^2/\mathcal A}{r^2}
 +w\left(\frac{8\pi Gp_r}{F}-\Phi'^2-\frac{4\Phi'}r\right)
 -\frac{4q\Phi'}r.
 \label{eq:radial-coefficients}
\end{align}
For a mode $e^{i\sigma t}\zeta(r)$, the complete radial equations reduce to
\begin{equation}
 \mathfrak t=P(\zeta'-Y\zeta),\qquad
 \mathfrak t'=-Y\mathfrak t-(Q+\sigma^2W)\zeta,\qquad \mathfrak t(a)=\mathfrak t(b)=0.
 \label{eq:radial-system}
\end{equation}
This is the elastic radial pulsation system~\cite{karlovini2004}, here written
with two free surfaces. Appendix~\ref{app:radial} derives it with the metric
perturbations retained. Spherical symmetry leaves no radiative gravitational
mode; those perturbations are reconstructed from the constraints.

\begin{theorem}[Coupled radial stability]
\label{thm:radial-stability}
Let a spherical Einstein--elastic equilibrium have a regular vacuum cavity,
two free faces at $0<a<b$, and smooth material coefficients with
$F,w,\mathcal A>0$ on $[a,b]$. If
\begin{equation}
 \inf_{[a,b]} e^{2\Phi-2\Lambda}\frac{\mathcal V}{w}\ge\gamma>0,
 \label{eq:radial-gap}
\end{equation}
then every radial eigenfrequency obeys $\sigma^2\ge\gamma$. Smooth radial
perturbations with the free-traction conditions have a conserved coercive
quadratic energy.
\end{theorem}
\begin{proof}
Put $\zeta=e^I y$, where $I(r)=\int_a^rY(s)\,ds$. Then
$-(Pe^{2I}y')'-Qe^{2I}y=\sigma^2We^{2I}y$, with $y'(a)=y'(b)=0$.
This is a regular Sturm--Liouville problem with positive leading coefficient
and weight, so its spectrum is real and discrete. Multiplication by $\zeta$ and integration gives
\begin{equation}
 \sigma^2\int_a^b W\zeta^2\,dr
 =\int_a^b\left[P(\zeta'-Y\zeta)^2-Q\zeta^2\right]dr
 \ge\gamma\int_a^bW\zeta^2\,dr.
 \label{eq:radial-rayleigh}
\end{equation}
The same integration for the time-dependent equation conserves
\begin{equation}
 E_2=\frac12\int_a^b
       \left[W\dot\zeta^2+P(\zeta'-Y\zeta)^2-Q\zeta^2\right]dr.
 \label{eq:radial-energy}
\end{equation}
Its boundary flux is $[\mathfrak t\dot\zeta]_a^b=0$. Positivity controls the
velocity and displacement; bounded $Y$ and positive $P$ also control one
spatial derivative.
\end{proof}

At zero coupling and relaxation,
$\mathcal V/w=12\kappa\mu_0/[m_0(\kappa+4\mu_0/3)r^2]>0$.
Continuity on the compact material interval therefore makes the analytic
branch radially stable for sufficiently small $G>0$.
For the numerical example, the sampled minimum of
$e^{2\Phi-2\Lambda}\mathcal V/w$ is positive, of order
$1.8\times10^{-3}L_0^{-2}$, illustrating the sufficient criterion.
Nonradial perturbations and rotating backgrounds require separate analyses.

\subsection{A nonlinear energy estimate}

In spherical comoving coordinates, let $V$ be the proper-time derivative of
areal radius and write the geometrized mass as $m=G\mathfrak m$. The exact
mass constraint defines the energy functional
\begin{equation}
 \mathfrak m_R=4\pi r^2r_R\,
 \rho\!\left(\frac{\sqrt{1+V^2-2G\mathfrak m/r}}{r_R},\frac Rr,\frac Rr\right),
 \quad \mathfrak m(A)=0,\qquad
 \mathcal E_G[r,V]=\mathfrak m(B).
 \label{eq:nonlinear-mass-functional}
\end{equation}
Here $\mathcal E_G=M/G$ is the physical ADM energy and $M$ the geometrized
ADM mass ($c=1$). The expression has a smooth $G\to0$ limit. It is conserved by
smooth spherical Einstein--elastic evolution with vacuum cavity and zero
outer traction.

\begin{proposition}[Conditional nonlinear radial control]
\label{prop:nonlinear-energy}
Fix the material annulus and law~\eqref{eq:elastic-law}. For an equilibrium
satisfying Theorem~\ref{thm:radial-stability}, there are
$c,C>0$ and a neighborhood $\mathcal U$ in $W^{1,\infty}\times L^\infty$
of $(r_*,0)$ such that
\begin{equation}
 c\|(r-r_*,V)\|_{\mathcal H}^2
 \le\mathcal E_G[r,V]-\mathcal E_G[r_*,0]
 \le C\|(r-r_*,V)\|_{\mathcal H}^2,
 \label{eq:nonlinear-coercivity}
\end{equation}
where $\|(u,v)\|_{\mathcal H}^2=\int_A^B(R^2u_R^2+2u^2+R^2v^2)\,dR$.
Every smooth spherical solution remaining in $\mathcal U$ therefore obeys
$\|(r(t)-r_*,V(t))\|_{\mathcal H}^2
\le(C/c)\|(r(0)-r_*,V(0))\|_{\mathcal H}^2$.
\end{proposition}
\begin{proof}
The mass functional is stationary at an equilibrium with free faces.
Its Hessian, for $\zeta=r^2e^{-\Phi}u$, is
\begin{equation}
 D^2\mathcal E_G[(u,v),(u,v)]
 =4\pi\int_a^b\left[P(\zeta'-Y\zeta)^2-Q\zeta^2+Wr^4v^2\right]dr.
 \label{eq:mass-hessian}
\end{equation}
The derivation is given in Appendix~\ref{app:radial}. The positive radial gap
makes this form equivalent to $\|(u,v)\|_{\mathcal H}^2$.
Smooth coefficients in the mass ODE, with $r,r_R$ and
$1+V^2-2G\mathfrak m/r$ bounded away from zero, give continuity of its Hessian
in this quadratic-form norm on a sufficiently small $\mathcal U$.
Taylor's formula gives Eq.~\eqref{eq:nonlinear-coercivity}; conservation gives
the evolution estimate.
\end{proof}

The neighborhood condition is essential. The energy controls one derivative
of the displacement in $L^2$, but not the $L^\infty$ strain or higher
derivatives needed for global smooth evolution. This estimate therefore does
not exclude shock formation or loss of regularity. General local
Einstein--elastic evolution requires compatible initial data at both
interfaces, including the higher-order gravitational and traction conditions
of Ref.~\cite{andersson2014}; local existence is not a stability theorem.

\begin{samepage}
\begin{proposition}[Compatible nontrivial radial data]
\label{prop:compatible-data}
Let a smooth spherical equilibrium have $r_R,F,w,\partial_{n_r}p_r>0$
up to both free faces. Every nonzero smooth radial velocity profile $v$
supported away from the faces generates a smooth family of spherical
Einstein--elastic initial data with
\begin{equation}
 V_\epsilon=\epsilon v,\qquad r_\epsilon=r_*+O(\epsilon^2),\qquad
 M_\epsilon=M_*+2\pi G\epsilon^2
     \int_a^b Wr^4v^2\,dr+o(\epsilon^2).
 \label{eq:compatible-data-mass}
\end{equation}
The data have the same material content, satisfy both Einstein constraints,
and agree with local static solutions near both interfaces. Their physical
traction and gravitational matching conditions are therefore compatible to
every order. Open constitutive and causal inequalities persist for small
$|\epsilon|$.
\end{proposition}
\end{samepage}
\begin{proof}
Appendix~\ref{app:radial} constructs the data by matching a static outer
collar through a scalar mass equation. Its derivative in the unknown outer
mass is $-1$ at equilibrium. The implicit function theorem and the mass
Hessian give~\eqref{eq:compatible-data-mass}. The explicitly reconstructed
three-metric and extrinsic curvature satisfy both constraints.
\end{proof}
These data supply examples for the following estimate. Static collars are
needed only to construct compatible initial data; the evolution may reflect
from either free face.

\begin{corollary}[Radial control through boundary reflections]
\label{cor:radial-short-time}
Fix a sufficiently weakly gravitating equilibrium, an integer $s\ge4$,
and the time normalization $\alpha(B)=1$. Let $z=(r,n_r,V,m)$.
For every finite $T>0$, there is $\epsilon_T>0$ such that smooth spherical
initial data satisfying both constraints and all interface compatibility
conditions, with $\|z(0)-z_*\|_{H^s}\le\epsilon<\epsilon_T$, have a unique
radial evolution on $[0,T]$, smooth up to each face from either side. Throughout this interval,
\begin{equation}
 \|z(t)-z_*\|_{H^s}\le C_0\epsilon e^{Ct}.
 \label{eq:radial-short-time}
\end{equation}
The constants depend on the background, $s$, and a fixed admissible
neighborhood, not on $\epsilon$. The data of
Proposition~\ref{prop:compatible-data} belong to this class.
\end{corollary}
\begin{proof}
Positive $\partial_{n_r}p_r$ allows radial pressure to replace radial density.
In material coordinates both boundary conditions become $p_r=0$.
Time differentiation preserves them, so every differentiated acoustic
energy has zero boundary flux. The field equations and propagated constraints
recover spatial derivatives without loss. Appendix~\ref{app:radial} gives
the resulting estimate $|\dot E_s|\le CE_s$, with
$E_s\asymp\|z-z_*\|_{H^s}^2$. Local construction, Gronwall, and continuation
then prove the statement: choose $\epsilon_T$ so that
$C_0\epsilon_T e^{CT}$ stays inside the fixed neighborhood.
\end{proof}
This controls arbitrarily many reflections on a prescribed finite interval.
The allowed amplitude depends on $T$; no fixed nonzero neighborhood is proved
to remain smooth for all time. In particular, the estimate gives neither
shock exclusion nor a lifetime bound for a specified disturbance.

\subsection{Nonlinear compression and wave breaking}
\label{sec:planar-steepening}

In a nonlinear elastic wave, the propagation speed depends on the local
state. Neighboring characteristics can approach one another and steepen a
compressive profile while its amplitude stays small. A planar limit of
the same material law makes this mechanism explicit before we restore
self-gravity and the two free faces.
At $G=0$, take material labels $(X(t,x),y,z)$, so that
$n_1=n=\sqrt{X_x^2-X_t^2}$, $n_2=n_3=1$, and $v=-X_t/X_x$.
The transverse equations vanish, and the longitudinal law is
\begin{equation}
 \rho_\parallel(n)=m_0n+\frac\kappa2(n-1)^2
               +\frac{\mu_0}{6}(n^3-2n+n^{-1}),\quad
 p_\parallel=n\rho_\parallel'-\rho_\parallel,\quad
 c^2(n)=\frac{n\rho_\parallel''}{\rho_\parallel'}.
 \label{eq:planar-law}
\end{equation}
Particle and momentum conservation diagonalize as
\begin{equation}
 R_\pm=\operatorname{artanh}v\pm\int_1^n\frac{c(s)}s\,ds,\qquad
 (\partial_t+\lambda_\pm\partial_x)R_\pm=0,
 \quad \lambda_\pm=\frac{v\pm c}{1\pm vc}.
 \label{eq:planar-invariants}
\end{equation}
These combinations are constant along their respective acoustic
characteristics and are called Riemann invariants~\cite{abbrescia2023shocks}.
For $R_-=0$, $v(n)=\tanh\int_1^n c(s)/s\,ds$ and
$\lambda_+'(1)=3c_0(1-c_0^2)/2=1/\sqrt3$ for the selected moduli.
A smooth density profile $N(\xi)$ close to one, equal to one outside a
finite interval, therefore has the exact solution
\begin{equation}
 x=\xi+t\lambda_+(N(\xi)),\qquad n=N(\xi),\qquad
 n_x=\frac{N'(\xi)}{1+t\lambda_+'(N(\xi))N'(\xi)}.
 \label{eq:planar-characteristics}
\end{equation}
If $a_*:=\min_\xi[\lambda_+'(N(\xi))N'(\xi)]<0$, the map remains
invertible for $t<-1/a_*$ and its first Jacobian zero produces an unbounded
gradient. The state values are transported unchanged. Particle conservation
reconstructs $X$ from $X_x=n/\sqrt{1-v^2}$ and $X_t=-vX_x$ until that time.
The following result retains self-gravity and the finite shell's two faces.

\begin{theorem}[Wave breaking inside a finite shell]
\label{thm:shell-steepening}
Fix $G>0$ and a smooth horizon-free spherical equilibrium of the selected elastic
law whose principal densities lie strictly inside
$[0.99,1.01]^3$. For every state-amplitude tolerance $\delta>0$ and time
$\tau>0$, there are smooth radial initial data with the same material
annulus, satisfying both Einstein constraints and all interface
compatibility conditions, whose classical evolution has a finite maximal
time $0<T_*<\tau$. The data are within $\delta$ of equilibrium in velocity,
principal densities, and metric coefficients in the uniform norm, using
material labels to compare the moving domains. Up to $T_*$ the
state remains in a fixed admissible neighborhood, but
\begin{equation}
 \limsup_{t\uparrow T_*}
 \|\partial_r(n_r,v)(t)\|_{L^\infty}=\infty.
 \label{eq:shell-gradient-breakdown}
\end{equation}
The singularity occurs in the material interior before a signal from the
pulse reaches either face. The exterior mass is allowed to change with
the data. No smallness of their higher Sobolev norms is asserted.
\end{theorem}
\begin{proof}
Appendix~\ref{app:shell-steepening} derives the exact radial characteristic
system and removes the interaction term from its differentiated outgoing
equation. The resulting weighted gradient satisfies
\begin{equation}
 \frac{dy}{dt}=-a_2y^2+a_1y+a_0,\qquad
 a_2\ge a_*>0,\quad |a_1|+|a_0|\le C,
 \label{eq:shell-riccati}
\end{equation}
along an outgoing acoustic characteristic. The coefficient bounds depend
on the state neighborhood, not on its gradients. Choose an interior
velocity pulse of amplitude $\varepsilon$ and label width $\ell$.
Solving the polar-slice mass constraint and adjusting a static outer
collar makes the data compatible; its adjustment is $O(\varepsilon^2\ell)$.
At a falling edge, $y(0)\le-c\varepsilon/\ell+O(1)$.

For $K>1+4C/a_*$, $y(0)\le-K$ implies
$y'\le-a_*y^2/2$ and preserves $y\le-K$. Hence
$-1/y(t)\le1/K-a_*t/2$, which precludes smooth continuation to
$2/(a_*K)$. The undifferentiated characteristic equations have bounded
sources, so their states remain in the chosen neighborhood for a time
independent of $\ell$. Take $\varepsilon$ small and then $\ell$ small
enough that this upper bound precedes $\tau$ and the travel time to either
face. The local continuation criterion then forces
Eq.~\eqref{eq:shell-gradient-breakdown}. The static collars remain regular
on this interval.
\end{proof}

The result applies in particular whenever the spherical example is
enclosed in this strain neighborhood. It supplies no numerical lifetime
for a specified disturbance. The ideal elastic action has no microscopic
length scale, viscosity, or plastic response: arbitrarily short pulses
test that continuum model. A physical prediction after wave breaking
requires additional material physics.

\section{Coupled axial energy}
\label{sec:axial-energy}

Odd-parity perturbations twist the material and couple to gravitational
radiation. Their energy must include the field in the cavity and exterior
as well as the elastic motion.
The axial equations of relativistic elasticity~\cite{karlovini2007} follow
from a quadratic action that makes this energy explicit.
First use an axisymmetric representative, and put
$\mathfrak f=r^2\sin^2\theta$. Write the perturbed metric and material
azimuth as
\begin{equation}
 g=\widehat g+\mathfrak f(d\phi+A_a dx^a)^2,\qquad
 \phi_{\rm mat}=\phi+\vartheta,\qquad
 \mathcal K_a=\partial_a\vartheta-A_a,\quad \mathcal F=dA,
 \label{eq:axial-fields}
\end{equation}
where $a\in\{t,r,\theta\}$ and
$\widehat g=-e^{2\Phi}dt^2+e^{2\Lambda}dr^2+r^2d\theta^2$ is the background
quotient metric. The infinitesimal axial gauge acts as
$(A,\vartheta)\mapsto(A+d\lambda,\vartheta+\lambda)$.
Let $w_t=\rho+p_t$, and let $c_r,c_\theta$ be the shear speeds for azimuthal
polarization and propagation along the orthonormal radial and polar directions.
Define
$C^{ij}=c_r^2 e_r^i e_r^j+c_\theta^2 e_\theta^i e_\theta^j$ and
$\mathcal S^{ab}=-u^au^b+C^{ab}$ in the material. Products containing $w_t$
vanish in vacuum.

The quadratic action, with the background field equations imposed, is
\begin{equation}
 S_2=-\int\sqrt{-g}\left[
 \frac{\mathfrak f}{64\pi G}\mathcal F_{ab}\mathcal F^{ab}
 +\frac{w_t\mathfrak f}{2}\mathcal S^{ab}\mathcal K_a\mathcal K_b\right]d^4x.
 \label{eq:axial-action}
\end{equation}
The first term follows from reduction along the spacelike axial Killing
field. For the second, the local material gradient changes by
$\delta F^3{}_i=n_t\sqrt{\mathfrak f}\,\mathcal K_i$, and its time derivative
by $n_t\sqrt{\mathfrak f}\,\mathcal K_0$. The action Hessian and
$n_t^2\mathsf W_{33}=w_t$ give the displayed coefficient. This includes the
background stress as well as the shear modulus. Variation gives
\begin{equation}
 \nabla_b(\mathfrak f\mathcal F^{ab})
   =16\pi G w_t\mathfrak f\mathcal S^{ab}\mathcal K_b,
 \qquad
 \nabla_a(w_t\mathfrak f\mathcal S^{ab}\mathcal K_b)=0.
 \label{eq:axial-equations}
\end{equation}
These agree with the covariant axial system of Ref.~\cite{karlovini2007}.

\begin{proposition}[Nonnegative coupled axial energy]
\label{prop:axial-energy}
Consider a static spherical elastic shell satisfying Einstein's equations,
with a regular vacuum cavity,
a horizon-free asymptotically flat exterior, and $w_t,c_r^2,c_\theta^2>0$
throughout the material. Smooth finite-energy axisymmetric axial perturbations satisfying
the linearized constraints and vacuum matching have the nonnegative energy
\begin{equation}
 E_{\rm ax}=\int_\Sigma\!\alpha\left[
 \frac{\mathfrak f}{32\pi G}\left(|\mathcal E|^2+\tfrac12|\mathcal B|^2\right)
 +\frac{w_t\mathfrak f}{2}
   \left(\mathcal K_0^2+C^{ij}\mathcal K_i\mathcal K_j\right)\right]dV_\gamma,
 \label{eq:axial-energy}
\end{equation}
where $\alpha=e^\Phi$, $\mathcal E_i=\mathcal F_{ti}/\alpha$,
$\mathcal B_{ij}=\mathcal F_{ij}$, and $\mathcal K_0=\mathcal K_t/\alpha$.
It is conserved when the flux at spatial infinity vanishes. If the null
radiation limit exists and satisfies the outgoing falloff below, its flux
at future null infinity is nonnegative.
There is no exponentially growing finite-energy mode with nonzero
$(\mathcal F,\mathcal K|_{\rm material})$.
\end{proposition}
\begin{proof}
The momenta of $A_i$ and $\vartheta$ are
$\Pi^i=\sqrt\gamma\,\mathfrak f\mathcal E^i/(16\pi G)$ and
$p=\sqrt\gamma\,w_t\mathfrak f\mathcal K_0$.
The Legendre transform of Eq.~\eqref{eq:axial-action} is the density in
Eq.~\eqref{eq:axial-energy}, plus $\Pi^i\partial_iA_t+pA_t$.
Integration by parts removes the latter terms by the Gauss constraint
$\partial_i\Pi^i=p$ and the asymptotically nonrotating boundary gauge.
All remaining coefficients are positive away from the regular axis.

At either material face, the torsional traction is
$w_t c_r^2\mathcal K_{\hat r}=0$. The metric matching conditions give
continuous tangential $A$ in a matching gauge and continuous normal flux
$\mathfrak f n_b\mathcal F^{ab}$. Thus the material energy flux vanishes
and the gravitational flux cancels across each interface. Regularity removes
axis and center terms. The remaining outward flux density is
\begin{equation}
 -\alpha^2\left[\frac{\mathfrak f}{16\pi G}
       \mathcal B^{ij}\mathcal E_j+w_t\mathfrak f C^{ij}\mathcal K_j\mathcal K_0
       \right]s_i.
 \label{eq:axial-flux}
\end{equation}
Here $s_i$ is the outward unit conormal in the background spatial metric.
Time translation invariance gives the energy balance. At fixed retarded
time, impose
\[
 \mathcal E_{\hat\theta}=O(r^{-2}),\qquad
 \mathcal B_{\hat r\hat\theta}=-\mathcal E_{\hat\theta}+o(r^{-2}),
\]
with regular angular profiles and an integrable radiation limit. The
material term is absent there, and the limiting gravitational flux is
nonnegative. For a complex mode, use the Hermitian extension of the energy;
exponential growth of a nonzero finite-energy field contradicts conservation.
Zero energy sets $\mathcal F=0$ and the material $\mathcal K=0$, leaving gauge
and constant label rotations. The argument uses no harmonic cutoff.
\end{proof}

Spherical symmetry removes the restriction to axisymmetric perturbations.
For normalized scalar harmonics $Y_{\ell m}$ on the unit sphere, define
$L_\ell=\ell(\ell+1)$, $\lambda_\ell=L_\ell-2$, and
$S_A^{\ell m}=\epsilon_A{}^B D_BY_{\ell m}/\sqrt{L_\ell}$ for $\ell\ge1$.
The vector and tensor normalizations are~\cite{martel2005}
\begin{equation}
 \int_{S^2}|S^{\ell m}|^2d\Omega=1,\qquad
 2\int_{S^2}|D_{(A}S_{B)}^{\ell m}|^2d\Omega=\lambda_\ell.
 \label{eq:axial-harmonic-norms}
\end{equation}
In Regge--Wheeler gauge, write
$\delta g_{aA}=r^2a_a(t,r)S_A^{\ell m}$ and the angular label displacement
as $q(t,r)S^{A}_{\ell m}$, with $a\in\{t,r\}$; the odd angular metric
perturbation vanishes. Put $e=\dot a_r-a_t'$. The corresponding modal energy is
\begin{align}
 E_{\ell m}=\int_0^\infty\!\alpha e^\Lambda r^4\Bigg\{
 &\frac{1}{32\pi G}\left[
 \frac{e^{-2\Lambda}}{\alpha^2}|e|^2+
 \frac{\lambda_\ell}{r^2}
   \left(\frac{|a_t|^2}{\alpha^2}+e^{-2\Lambda}|a_r|^2\right)\right]
 \nonumber\\[-1mm]
 &+\frac{w_t}{2}\left[
 \frac{|\dot q-a_t|^2}{\alpha^2}
 +c_r^2e^{-2\Lambda}|q'-a_r|^2
 +\frac{c_\theta^2\lambda_\ell}{r^2}|q|^2\right]\Bigg\}\,dr.
 \label{eq:axial-modal-energy}
\end{align}
The material terms are integrated only over the shell. For $\ell=1$ the
odd tensor harmonic vanishes; the remaining axial gauge is quotiented.

\begin{corollary}[All odd-parity harmonics]
\label{cor:axial-all-m}
Under the background hypotheses of Proposition~\ref{prop:axial-energy},
smooth odd-parity solutions with finite
$E_{\rm odd}=\sum_{\ell\ge1}\sum_{m=-\ell}^{\ell}E_{\ell m}$
obey the same energy and flux conclusions, provided the stated boundary
and radiation conditions hold with finite total flux. In particular,
there is no exponentially growing finite-energy physical odd-parity mode.
\end{corollary}
\begin{proof}
The unit-sphere identities $D_AS^A=0$ and
$D^BD_BS_A=(1-L_\ell)S_A$ give Eq.~\eqref{eq:axial-harmonic-norms}
by integration by parts, using $\operatorname{Ric}_{AB}=\Omega_{AB}$.
For $m=0$, put $Z=S^\phi$, so that
$(A_t,A_r,\vartheta)=(a_t,a_r,q)Z$ and $A_\theta=0$.
Then $\int\sin^2\theta\,|Z|^2d\Omega=1$ and
$\int\sin^2\theta\,|Z'|^2d\Omega=\lambda_\ell$; substitution in
Eq.~\eqref{eq:axial-energy} gives Eq.~\eqref{eq:axial-modal-energy}.
The background action, constraints, and both interfaces are rotation
invariant. Orthogonality of vector and tensor harmonics therefore makes
the quadratic action diagonal in $(\ell,m)$ with these same coefficients
for every $m$. This also transfers the modewise flux identity.
Since $\lambda_\ell=(\ell-1)(\ell+2)\ge0$, every summand is nonnegative.
Conservation passes to the convergent sum, and monotone convergence applies
to the nonnegative outgoing flux. A nonzero growing mode would make at
least one of these conserved nonnegative modal energies grow. Completeness
of the harmonic expansion leaves only the gauge and label zero modes.
\end{proof}

The constitutive bound~\eqref{eq:material-box-bounds} supplies the material
hypotheses whenever the spherical equilibrium remains in the strain box. The $\ell=1$ sector includes changes in angular momentum and has no
radiative gravitational mode; a comparison at fixed angular momentum removes
that stationary variation. There is no axial $\ell=0$ mode. The estimate
retains the regular cavity and the unbounded exterior. It does not address
even-parity modes, a rotating background, or nonlinear axial evolution,
and supplies no positive frequency gap or decay rate.

In the exterior vacuum, the radiative variable obeys the Regge--Wheeler equation.
For $\ell\ge2$, put $\Psi=r^3e$ in the vacuum part of the modal action.
Its two Euler--Lagrange equations give
\begin{equation}
 a_t=-\frac{F(r\Psi)'}{\lambda_\ell r^2},\qquad
 a_r=-\frac{\dot\Psi}{\lambda_\ell Fr},\qquad
 \ddot\Psi-\partial_{r_*}^2\Psi+V_\ell\Psi=0,
 \quad \frac{dr_*}{dr}=F^{-1},
 \label{eq:axial-vacuum-reconstruction}
\end{equation}
where $V_\ell=F[\ell(\ell+1)/r^2-6M/r^3]$.
Substitution in the energy and one integration by parts give the usual
Regge--Wheeler energy, up to a positive normalization, for compact exterior
data. Thus the radiative mode has been retained. The reconstruction divides
by $\lambda_\ell$ and does not apply to the dipole.

\section{Radial constraints and the constrained mass}
\label{app:radial}

All variations in this appendix preserve the material annulus, its reference
metric, and its particle content. Primes denote areal-radius derivatives on
the background; $R$ remains the particle label. Put
$\mathcal C=n_r\partial_{n_r}p_t$. Commutation of the mixed derivatives of
$\rho(n_r,n_t,n_t)$ gives
\begin{equation}
 \mathcal B=2\mathcal C-2q.
 \label{eq:elastic-reciprocity}
\end{equation}
This identity is needed for self-adjoint radial dynamics; independently
prescribed anisotropic pressure perturbations need not satisfy it.

In polar areal coordinates, material kinematics gives
$\Delta n_r/n_r=-\xi'-\Delta\Lambda$ and
$\Delta n_t/n_t=-\xi/r$. Before eliminating the metric,
$\delta\Lambda=\delta m/(rF)$. Differentiating the mass constraint and using
$p_r'=-w\Phi'+2q/r$ yields
\begin{equation}
 D_m'+\frac{4\pi G r w}{F}D_m=0,\qquad
 D_m=\delta m+4\pi G r^2w\xi.
 \label{eq:mass-constraint-propagation}
\end{equation}
Vacuum matching at the moving inner face gives $\Delta m(a)=0$.
Since $p_r(a)=0$, this is $D_m(a)=0$. Hence $D_m=0$ throughout,
which proves Eq.~\eqref{eq:radial-constraints}. It also gives
$\Delta m=-4\pi G r^2p_r\xi$, so $\Delta m(b)=0$ at the outer face.
The exterior mass has no linear perturbation.

The remaining radial Einstein constraint is
\begin{equation}
 \delta\Phi'=\frac{4\pi G r}{F}\delta p_r
 -\frac{4\pi G w}{F}(1+2r\Phi')\xi.
 \label{eq:lapse-perturbation}
\end{equation}
Normalization at infinity fixes its integration constant. Radial momentum
conservation, including its inertial term, reads
\begin{equation}
 w e^{2\Lambda-2\Phi}\ddot\xi+\delta p_r'
 +(\delta\rho+\delta p_r)\Phi'+w\delta\Phi'
 -\frac2r(\delta p_t-\delta p_r)=0.
 \label{eq:full-radial-force}
\end{equation}
Let $\Pi=\Delta p_r$. Substituting the constitutive variations, the
constraints, and Eq.~\eqref{eq:elastic-reciprocity} into the last four terms
gives exactly
\begin{equation}
 \delta\mathcal F=\Pi'+(\Lambda'+2\Phi'+Y)\Pi+\mathcal V\xi,
 \qquad \Pi=-e^{-\Lambda-2\Phi}\mathfrak t.
 \label{eq:radial-force-reduction}
\end{equation}
Here $\mathcal F=p_r'+w\Phi'-2q/r$ is the static force residual.
Equation~\eqref{eq:radial-system} follows, with the elastic pulsation
coefficients of Ref.~\cite{karlovini2004}; its anisotropy variable is $q/3$
in our notation. The derivation retains both metric perturbations.

For the nonlinear energy, use comoving spherical coordinates. Then
$n_r=e^{-\lambda}$ and
$e^{-2\lambda}r_R^2=1+V^2-2m/r$. The exact Einstein equations give
$m_R=4\pi G r^2\rho r_R$ and
$u(m)=-4\pi G r^2p_r V$~\cite{karlovini2004}.
They imply Eq.~\eqref{eq:nonlinear-mass-functional} and conservation at the
free outer surface. The inner vacuum mass remains zero.

To vary the endpoint mass, write its geometrized ODE as $m_R=f(R,r,r_R,V,m)$.
At $V=0$, direct differentiation gives
\begin{equation}
 f_{r_R}=-4\pi G r^2p_r,\qquad
 f_m=-\frac{4\pi G r r_R w}{F}.
 \label{eq:mass-ode-derivatives}
\end{equation}
The adjoint weight solves $\chi_R=-f_m\chi$, $\chi(B)=1$.
At equilibrium it is $\chi=e^{\Phi+\Lambda}$, because
$(\Phi+\Lambda)'=4\pi G r w/F$ and $\Phi(b)+\Lambda(b)=0$.
For a material displacement $u(R)$, integration by parts gives
\begin{equation}
 D\mathcal E_G[u,0]
 =-4\pi[\chi r^2p_ru]_A^B
 +4\pi\int_A^B\chi r^2r_R\mathcal F u\,dR=0.
 \label{eq:mass-first-variation}
\end{equation}
Thus both free-traction conditions are natural boundary conditions of the
mass functional.

Differentiate Eq.~\eqref{eq:mass-first-variation} at equilibrium. Terms
multiplying $p_r$ at the faces or $\mathcal F$ in the bulk vanish.
With $\zeta=r^2e^{-\Phi}u$, the remaining boundary term is
$4\pi[\zeta\mathfrak t]_a^b$. Equation~\eqref{eq:radial-force-reduction} reduces
the volume term to
$-4\pi\int_a^b\zeta(\mathfrak t'+Y\mathfrak t+Q\zeta)\,dr$.
Its integration by parts cancels that boundary term and leaves
$4\pi\int[P(\zeta'-Y\zeta)^2-Q\zeta^2]dr$.
This calculation allows arbitrary endpoint displacements and does not
assume the linearized traction conditions on the variations.

The mass ODE is even in $V$, so mixed displacement--velocity terms vanish
at $V=0$. Its velocity Hessian is
$f_{VV}=4\pi G r^2r_Rw/F$. Multiplication by $\chi/G$ and conversion to
$dr$ give the kinetic term $4\pi Wr^4v^2$ in
Eq.~\eqref{eq:mass-hessian}. At $G=0$ and relaxation the full
Hessian reduces to
\begin{equation}
 4\pi\int_A^B R^2\left[
 \kappa(u_R+2u/R)^2+\frac{4\mu_0}{3}(u_R-u/R)^2+m_0v^2\right]dR.
 \label{eq:relaxed-mass-hessian}
\end{equation}
The elastic part is bounded below by
$c_0(u_R^2+2u^2/R^2)$ for any
$0<c_0\le\min(3\kappa,2\mu_0)$.

Finally, choose a convex $W^{1,\infty}\times L^\infty$ neighborhood on which
$r,r_R$ and $1+V^2-2G\mathfrak m/r$ have positive lower bounds.
The mass ODE and its first two variational equations then have uniformly
bounded smooth coefficients on the compact material interval. Gronwall's
inequality bounds the first mass variation in $L^\infty$ by the
$H^1\times L^2$ variation norm. The second variational equation contains
only products of these first variations and of $(u,u_R,v)$.
Cauchy--Schwarz therefore bounds its endpoint by the squared norm.
The same estimates for differences of coefficients give continuity of the
Hessian in quadratic-form norm as the base state varies in this neighborhood.
The equilibrium Hessian is coercive by Eq.~\eqref{eq:radial-gap}; shrinking
the neighborhood preserves coercivity. Taylor's integral remainder proves
Eq.~\eqref{eq:nonlinear-coercivity}. This argument supplies no estimate of
higher derivatives of an evolving solution.

\subsection{Construction of compatible radial initial data}
Write $\Gamma^2=1+V^2-2m/r$. On an initial material slice set
\begin{equation}
 \gamma=\frac{r_R^2}{\Gamma^2}dR^2+r^2d\Omega^2,\qquad
 K^R{}_R=-\frac{V_R}{r_R},\qquad
 K^\theta{}_\theta=K^\phi{}_\phi=-\frac Vr.
 \label{eq:radial-initial-geometry}
\end{equation}
The normal is the material velocity, so the momentum density vanishes.
Direct calculation gives
\begin{equation}
 D_j(K^j{}_i-\delta^j_i K)=0,\qquad
 {}^{(3)}R+K^2-K_{ij}K^{ij}=\frac{4m_R}{r^2r_R}.
 \label{eq:radial-initial-constraints}
\end{equation}
Thus the mass ODE~\eqref{eq:nonlinear-mass-functional}, with
$n_r=\Gamma/r_R$ and $n_t=R/r$, solves the Hamiltonian constraint as well.

Fix the equilibrium outer radius $b$ and let $\mathsf M$ vary near $M_*$.
Solve the static material ODE backward on an outer collar, with terminal
data $r(B)=b$, $m(B)=\mathsf M$, and $p_r(B)=0$. Positive
$\partial_{n_r}p_r$ determines the terminal radial density and guarantees
a smooth family of collar solutions $r^{\mathsf M}_{\rm out}$.
Choose a smooth cutoff $\chi$ supported in that collar and equal to one
on a smaller collar adjoining $B$, with support disjoint from that of $v$.
Define
\begin{equation}
 r_{\mathsf M}=r_*+\chi(r^{\mathsf M}_{\rm out}-r_*),\qquad
 H(\mathsf M,\epsilon)=G\mathcal E_G[r_{\mathsf M},\epsilon v]-\mathsf M.
 \label{eq:collar-mass-matching}
\end{equation}
At $(M_*,0)$, $H=0$, $H_{\mathsf M}=-1$, and $H_\epsilon=0$:
the derivatives of $\mathcal E_G$ vanish at the free equilibrium for
every displacement and velocity. Hence $H(\mathsf M(\epsilon),\epsilon)=0$
has a unique local smooth solution with $\mathsf M'(0)=0$.
The second derivative is the kinetic Hessian in~\eqref{eq:mass-hessian},
which proves~\eqref{eq:compatible-data-mass}.

Integrate the mass constraint from the inner face. Near that face its
data are the original static solution. On the outer collar, both this
mass function and the local static one obey the same first-order ODE
with the same terminal mass, so uniqueness makes them equal there.
The reconstructed initial data are therefore static near each face and
match flat and Schwarzschild vacuum, respectively. Local static extensions
give all time derivatives required for interface compatibility; lapse
normalizations on the two collars may differ by constant time rescalings.
An initial harmonic gauge can be chosen for these geometrically compatible
data. This construction does not hold the exterior mass fixed while changing
the interior energy.

The exact radial evolution also displays the missing derivative control.
With $D_t=\alpha^{-1}\partial_t$, $\nu=n_r$, $\eta=R/r$, and
$q=p_t-p_r$, its comoving equations are
\begin{equation}
 \begin{gathered}
 D_t r=V,\qquad D_t\nu=-\frac{\nu^2}{\Gamma}V_R,\qquad
 D_t m=-4\pi G r^2p_rV,\\
 D_tV=-\frac{\Gamma\nu}{w}(p_r)_R
       +\frac{2\Gamma^2q}{wr}-\frac{m}{r^2}-4\pi G rp_r,\\
 (\log\alpha)_R=-\frac{(p_r)_R}{w}+
                         \frac{2q r_R}{wr}.
 \end{gathered}
 \label{eq:nonlinear-radial-evolution}
\end{equation}
The identity $r_R=\Gamma/\nu$ and the mass constraint are propagated.
Direct spacetime curvature gives
$G^R{}_R=-2D_tV/r-2m/r^3+2\Gamma\nu(\log\alpha)_R/r$;
the mixed Einstein equation gives $D_t\nu=-\nu V_R/r_R$.
These identities recover the displayed evolution from the radial Einstein
equation and stress conservation. With $r_R>0$, the angular equation follows
from the radial contracted Bianchi identity. Thus no independent angular
Einstein equation is omitted.
For the pair $(\nu,V)$ the principal matrix and a symmetrizer are
\begin{equation}
 A=\alpha\begin{pmatrix}0&\nu^2/\Gamma\\
     \Gamma\nu\partial_\nu p_r/w&0\end{pmatrix},\qquad
 S=\operatorname{diag}\!\left(\frac{\Gamma^2\partial_\nu p_r}{\nu w},1\right).
 \label{eq:radial-evolution-symbol}
\end{equation}
Both $S$ and the squared coordinate characteristic speed
$\alpha^2\nu^3\partial_\nu p_r/w$ are positive on the admissible branch.
The local proper radial distance is $d\ell=dR/\nu$ and proper time is
$d\tau=\alpha\,dt$, so the physical squared speed is
$\nu\partial_\nu p_r/w$, in agreement with the material acoustic test.
Differentiating the traction condition once gives
$\mathcal A V_R/r_R+\mathcal B V/r=0$ at each face.
Higher derivatives of this condition are supplied by the static collars
at the initial time. A nonlinear continuation estimate must control the
coefficients and differentiated equations in~\eqref{eq:nonlinear-radial-evolution};
the conserved endpoint mass alone does not provide that estimate.

For the boundary estimate, first eliminate $r_R$ and $\eta_R$ from the
force equation using $r_R=\Gamma/\nu$ and
$\eta_R=1/r-\eta\Gamma/(\nu r)$. The lapse introduces no spatial derivative
loss. Indeed, define
$h(\nu,\eta)=\int_1^\nu (\partial_a p_r)/(\rho+p_r)\,da$.
Then
\begin{equation}
 (\log\alpha+h)_R=
 \left(h_\eta-\frac{\partial_\eta p_r}{w}\right)
 \left(\frac1r-\frac{\eta\Gamma}{\nu r}\right)
 +\frac{2q\Gamma}{\nu wr}.
 \label{eq:radial-lapse-no-loss}
\end{equation}
The right side contains only the state. Integration from $B$, with
$\alpha(B)=1$, makes $z\mapsto\alpha[z]$ locally Lipschitz in $H^s$ on
an admissible neighborhood, including its boundary trace.

The pressure formulation below preserves the constraints, including
off-constraint errors. Define $C_1=r_R-\Gamma/\nu$ and
$C_2=m_R-4\pi G r^2\rho r_R$. Direct differentiation, using the integrated
lapse and $[D_t,\partial_R]= (\log\alpha)_R D_t$, gives
\begin{align}
 D_tC_1&=\frac{\eta V}{r}\left(h_\eta-\frac{\partial_\eta p_r}{w}\right)C_1,\nonumber\\
 D_tC_2&=4\pi G r^2\left[\frac{w\nu}{\Gamma}V_R+
 \frac{V}{r}(2q+\eta\partial_\eta p_r-\eta w h_\eta)\right]C_1.
 \label{eq:radial-constraint-errors}
\end{align}
These homogeneous equations keep initially vanishing constraints zero.

\subsection{An estimate that includes reflections}
Use $p=p_r$ as a variable, writing $\nu=N(p,\eta)$ by
$\partial_\nu p_r>0$. Subscripts $t$ and $R$ on $p$ in this subsection
denote derivatives. Put
\begin{equation}
 a_p=\frac{\nu^2\partial_\nu p_r}{\Gamma},\qquad
 b_p=\frac{\Gamma\nu}{w},\qquad
 f_p=-\frac{\eta\partial_\eta p_r}{r}V,\qquad
 f_V=\frac{2\Gamma^2q}{wr}-\frac{m}{r^2}-4\pi G rp.
 \label{eq:pressure-coefficients}
\end{equation}
Here the material derivative $\partial_\eta p_r$ is taken at fixed $\nu$.
The exact evolution becomes
\begin{align}
 p_t+\alpha a_p V_R&=\alpha f_p,&
 V_t+\alpha b_p p_R&=\alpha f_V,\nonumber\\
 r_t&=\alpha V,&m_t&=-4\pi G\alpha r^2pV,
 \qquad p(t,A)=p(t,B)=0.
 \label{eq:pressure-evolution}
\end{align}
The acoustic symmetrizer is $\operatorname{diag}(a_p^{-1},b_p^{-1})$;
its flux is $\alpha PV$ for pressure and velocity variations $(P,V)$.
There is one incoming characteristic at each face, with speeds
$\pm\alpha\sqrt{a_pb_p}$. The condition $P=0$ relates the two
characteristic amplitudes by reflection and has zero energy flux. No
boundary condition is needed for the two ordinary differential equations.

Let $P_j=\partial_t^j(p-p_*)$ and $V_j=\partial_t^j V$.
Every $P_j$ vanishes at both faces. Multiplication of the time-differentiated
acoustic equations by $(P_j/a_p,V_j/b_p)$ therefore leaves no boundary
term, at any order. Take
\begin{equation}
 E_s=\frac12\sum_{j=0}^s\int_A^B
       \left(\frac{P_j^2}{a_p}+\frac{V_j^2}{b_p}\right)dR
       +\frac12\|(r-r_*,m-m_*)\|_{L^2}^2.
 \label{eq:radial-higher-energy}
\end{equation}
On a bounded admissible $H^s$ neighborhood this energy is equivalent to
$\|z-z_*\|_{H^s}^2$. To see why no derivatives are missing, use
\begin{equation}
 \begin{aligned}
 V_R&=\frac{f_p-p_t/\alpha}{a_p},&
 p_R&=\frac{f_V-V_t/\alpha}{b_p},\\
 r_R&=\frac{\Gamma}{N(p,\eta)},&
 m_R&=4\pi G r^2\rho\frac{\Gamma}{N(p,\eta)}.
 \end{aligned}
 \label{eq:radial-derivative-recovery}
\end{equation}
After subtracting the static equations, induction recovers each mixed
space--time derivative of total order at most $s$ from the pure time
energies. The last two equations recover one spatial derivative of $r,m$
from lower-order state derivatives. Conversely, the evolution recovers
pure time derivatives from the spatial $H^s$ norm.

The lapse boundary trace needs attention in this induction. At $B$,
$p=0$ implies $\nu=N(0,B/r(B))$. Hence the endpoint value of $h$ in
Eq.~\eqref{eq:radial-lapse-no-loss} is a smooth function of $r(B)$ alone.
Since $r_t(B)=V(B)$, its $j$th time derivative requires at most the trace
of $\partial_t^{j-1}V$. Its $H^1$ norm is controlled by
Eq.~\eqref{eq:radial-derivative-recovery} at order $j$.
The integral part of the lapse is bounded in $L^2$ by its integrand's
$L^2$ norm. Thus $\partial_t^j\alpha$ is controlled at the same order,
without a trace of $\partial_t^j p$ or $\partial_t^j V$.

The one-dimensional product estimates now bound every commutator.
The top terms have the form $\partial_t^j a_p\,V_R$,
$\partial_t^j b_p\,p_R$, or $\partial_t^j\alpha$ times a first
state derivative; that first derivative is bounded in $L^\infty$ for
$s\ge4$. The remaining products have lower total order and the same
bound. The equilibrium residual vanishes, so all forcing terms in the
difference equations contain a perturbation. Differentiation of the
weights and of the last two terms in~\eqref{eq:radial-higher-energy}
therefore gives
\begin{equation}
 |\dot E_s|\le C E_s.
 \label{eq:radial-reflection-estimate}
\end{equation}
The constants are uniform in a fixed bounded neighborhood with
$r,\nu,\Gamma,w,\partial_\nu p_r$ bounded away from zero.

The construction must also estimate fields that do not yet satisfy the
constraints. For this purpose replace the last term of
Eq.~\eqref{eq:radial-higher-energy} by
$\tfrac12\|(r-r_*,m-m_*)\|_{H^s}^2$, and call the resulting energy
$\widetilde E_s$. The first two equations of
Eq.~\eqref{eq:radial-derivative-recovery} recover derivatives of $p,V$
without using either constraint. Together with the retained $H^s$ norms
of $r,m$, they give
$\widetilde E_s\asymp\|(p-p_*,V,r-r_*,m-m_*)\|_{H^s}^2$.
The integrated lapse is a smooth state functional on this space. In
particular, its endpoint value involves $r(B)$ through $p(B)=0$;
its differentiated traces are controlled by the acoustic recovery equations
as above. The $H^s$ product estimate applied to
$r_t=\alpha V$ and $m_t=-4\pi G\alpha r^2pV$ bounds the derivative of
the added norms by $C\widetilde E_s$. Thus the energy estimate closes
without imposing constraints on intermediate fields.

Freeze the coefficients in~\eqref{eq:pressure-evolution}, retain the
homogeneous pressure conditions, and solve the acoustic characteristics
and the $r,m$ equations with compatible initial time jets. The preceding
estimate bounds the iterates in a fixed $H^s$ neighborhood. The difference
system has the same zero pressure trace; the state-to-lapse map is locally
Lipschitz in $H^{s-1}$, and the product estimate at that order gives
contraction for sufficiently short time. This constructs the radial solution
and permits continuation while its $H^s$ norm stays bounded and its state
remains admissible. The assumed interface compatibility removes corner
terms. Only after taking the limit do
Eqs.~\eqref{eq:radial-constraint-errors} propagate the initially vanishing
constraints. Their use then reduces $\widetilde E_s$ to the equivalent
constrained energy $E_s$. At either face $p=0$ keeps the vacuum mass constant; the
induced metric and extrinsic curvature match the regular cavity and the
Schwarzschild exterior at the moving face. The initial geometric
compatibility assumed in Corollary~\ref{cor:radial-short-time} is the one
in Ref.~\cite{andersson2014}. The argument does not require static collars
after the initial slice. Gronwall and Sobolev embedding close the bootstrap
on any fixed finite interval at sufficiently small amplitude.

\section{Interior wave breaking with self-gravity}
\label{app:shell-steepening}

\subsection{A radial system with bounded sources}

Use polar areal coordinates while $F=1-2m/r>0$,
\begin{equation}
 ds^2=-\alpha^2dt^2+a^2dr^2+r^2d\Omega^2,\qquad a=F^{-1/2},
 \qquad \alpha\to1\quad(r\to\infty).
 \label{eq:polar-evolution-metric}
\end{equation}
Here $a(t,r)$ is a metric coefficient, not the inner face radius.
Let $v=\tanh\vartheta$ be the radial velocity measured by the slice
normal, $\gamma_v=(1-v^2)^{-1/2}$, $\nu=n_r$, $\eta=R/r$,
$p=p_r$, $q=p_t-p$, and $w=\rho+p$. The normal-frame source components are
\begin{equation}
 E=w\gamma_v^2-p,\qquad j=w\gamma_v^2v,\qquad
 P=w\gamma_v^2v^2+p.
 \label{eq:polar-sources}
\end{equation}
The radial Einstein equations and the material-label identity give
\begin{align}
 m_r&=4\pi G r^2E,&m_t&=-4\pi G r^2\alpha j/a,
 & (\log\alpha)_r&=a^2(m/r^2+4\pi G rP)=:L,\nonumber\\
 \eta_r&=(a\nu\gamma_v-\eta)/r,&
 \eta_t&=-\alpha\nu\gamma_v v/r,&
 a_t/(\alpha a)&=-4\pi G ra j=:B_g.
 \label{eq:polar-constraints-evolution}
\end{align}
In particular, the first derivatives of $m,\eta$ and the spatial derivative
of the lapse depend on state values only. 

Define $D=\alpha^{-1}\partial_t+v a^{-1}\partial_r$ and
$S=v\alpha^{-1}\partial_t+a^{-1}\partial_r$. Particle conservation and
the radial projection of stress conservation become
\begin{align}
 D\log\nu+S\vartheta&=F_1:=-vL/a-B_g,\nonumber\\
 D\vartheta+c^2S\log\nu&=F_2:=-L/a-vB_g+\frac{2q}{war}
                  -\frac{p_\eta}{w}\frac{\nu/\gamma_v-\eta/a}{r},
 \qquad c^2=\frac{\nu p_\nu}{w}.
 \label{eq:polar-acoustic-system}
\end{align}
Indeed, $D\eta=-v\eta/(ar)$ and
$S\eta=(\nu/\gamma_v-\eta/a)/r$; these remove the tangential-density
and spherical-divergence terms from the principal part. Hyperelasticity
makes the energy equation a consequence of these kinematics and the
stress law. Equations~\eqref{eq:polar-constraints-evolution} and
\eqref{eq:polar-acoustic-system} retain the full radial Einstein--elastic
system, including the tangential stress.

Set
\begin{equation}
 I(\nu,\eta)=\int_1^\nu\frac{c(s,\eta)}s\,ds,\quad
 U_\pm=\vartheta\pm I(\nu,\eta),\quad
 \lambda_\pm=\alpha\ell_\pm,\quad
 \ell_\pm=\frac{v\pm c}{a(1\pm vc)}.
 \label{eq:shell-characteristic-variables}
\end{equation}
Positive sound speed makes this a smooth change of the two acoustic
variables at fixed $\eta$. Direct substitution gives
\begin{equation}
 (\partial_t+\lambda_\pm\partial_r)U_\pm=\alpha f_\pm,
 \qquad
 f_\pm=\frac{F_2\pm cF_1\pm I_\eta(D\eta\pm cS\eta)}{1\pm vc}.
 \label{eq:shell-characteristic-sources}
\end{equation}
The sources are smooth functions of $(U_+,U_-,\eta,m,r)$ on any compact
admissible state set bounded away from $r=0$ and $F=0$.

The outgoing characteristic speed changes strictly with the outgoing
wave amplitude throughout the strain box. This property, called genuine
nonlinearity, is what permits compression to steepen.
At fixed $\eta$, write $\rho_\nu=\partial_\nu\rho$ and similarly for
higher derivatives. The explicit law gives
\begin{equation}
 \frac{\rho_{\nu\nu}}{m_0}
 =\frac{6\nu+12\eta^4+2\eta^4/\nu^3}{60}>\frac14,
 \qquad
 \frac{\nu\rho_{\nu\nu\nu}}{m_0}
 =\frac{\nu-\eta^4/\nu^3}{10},\qquad
 |\nu\rho_{\nu\nu\nu}|<\frac{m_0}{100}.
 \label{eq:shell-density-derivatives}
\end{equation}
For example, bound $\eta^4/\nu^3$ between
$(99/100)^4/(101/100)^3$ and $(101/100)^4/(99/100)^3$;
substitution gives both displayed bounds. Since
$c^2=\nu\rho_{\nu\nu}/\rho_\nu\le2/5$,
\begin{equation}
 \mathfrak g:=1-c^2+\frac{\nu c_\nu}{c}
 =\frac32(1-c^2)+\frac{\nu\rho_{\nu\nu\nu}}{2\rho_{\nu\nu}}
 >\frac{22}{25},\qquad
 \partial_{U_+}\ell_+
 =\frac{(1-v^2)\mathfrak g}{2a(1+vc)^2}>0.
 \label{eq:shell-genuine-nonlinearity}
\end{equation}
The derivative holds $U_-,\eta,m,r$ fixed. At $G=0$, $\nu=\eta=1$ and $v=0$, it is $1/2$.
This positive coefficient drives the steepening of a falling profile.

\subsection{Removing the other characteristic gradient}

We give the reduction explicitly because bounded sources by themselves
would not justify a scalar blowup comparison. Write $u=U_+$, $w_-=U_-$,
$z=(\eta,m,r)$ and $d=\ell_+-\ell_->0$. Equations
\eqref{eq:polar-constraints-evolution} determine smooth functions
$g(z,u,w_-)=z_r$ and $b(z,u,w_-)=z_t/\alpha$, with the last components
$1$ and $0$. Thus
$\alpha^{-1}(\partial_t+\lambda_+\partial_r)z=b+\ell_+g=:g_+$.
All derivatives below in $u,w_-,z$ are ordinary state derivatives.
Choose smooth primitives on a small rectangular state neighborhood,
\begin{equation}
 h_{w_-}=\frac{(\ell_+)_{w_-}}d,\qquad
 k_{w_-}=-\frac{e^h(f_+)_{w_-}}d,
 \qquad y=e^h\partial_r u+k.
 \label{eq:shell-gradient-transform}
\end{equation}
Integration in $w_-$ at fixed $(u,z)$ defines $h$ and then $k$; their
integration constants may be zero. Neither depends on the lapse.
Put
\begin{align}
 B_0&=(f_+)_u-(\ell_+)_z\cdot g-L\ell_+,&
 C_0&=(f_+)_z\cdot g+Lf_+,\nonumber\\
 B_1&=B_0+h_u f_++h_{w_-}f_-+h_z\cdot g_+,&
 C_1&=e^h C_0+k_u f_++k_{w_-}f_-+k_z\cdot g_+.
 \label{eq:shell-riccati-coefficients}
\end{align}
Differentiating Eq.~\eqref{eq:shell-characteristic-sources} now gives exactly
\begin{equation}
 (\partial_t+\lambda_+\partial_r)y
 =-\alpha e^{-h}(\ell_+)_u(y-k)^2+\alpha B_1(y-k)+\alpha C_1.
 \label{eq:shell-exact-riccati}
\end{equation}
The $u_r$ equation contains the terms
$-\alpha(\ell_+)_{w_-}u_r(w_-)_r$ and
$\alpha(f_+)_{w_-}(w_-)_r$. The identity
\[
 (\partial_t+\lambda_+\partial_r)w_-=\alpha f_-+\alpha d(w_-)_r
\]
shows how the two primitives in~\eqref{eq:shell-gradient-transform} cancel them
in turn. No derivative of $\alpha$ in time is taken. This avoids estimating
a time derivative of a nonlocal lapse integral by an unbounded strain
gradient.

On a compact state neighborhood, the lapse is bounded above and below by
integrating $L$ through the finite shell and matching to the exterior.
Equation~\eqref{eq:shell-genuine-nonlinearity} therefore bounds the
quadratic coefficient away from zero; all other coefficients are bounded
independently of spatial gradients. Expansion of
Eq.~\eqref{eq:shell-exact-riccati} gives~\eqref{eq:shell-riccati}.
This is the usual characteristic steepening argument
\cite{abbrescia2023shocks}, with the geometric and elastic source terms
retained explicitly.

\subsection{Compatible pulses and continuation}

A polar initial slice differs from the comoving slice in
Eq.~\eqref{eq:nonlinear-mass-functional}. Prescribe an increasing embedding
$r(R)$ and a velocity $v(R)$ supported away from the faces. The polar
Hamiltonian constraint is the regular scalar ODE
\begin{equation}
 \begin{gathered}
 m_R=4\pi G r^2r_R\left(\frac{\rho+p}{1-v^2}-p\right),
 \qquad m(A)=0,\\
 \nu=\frac{\sqrt{1-2m/r}\sqrt{1-v^2}}{r_R},\qquad
 \eta=R/r.
 \end{gathered}
 \label{eq:polar-initial-mass}
\end{equation}
The spatial metric is $a^2dr^2+r^2d\Omega^2$. With our sign convention,
$K^r{}_r=4\pi G ra j$ and $K^\theta{}_\theta=K^\phi{}_\phi=0$ solve
the momentum constraint; $K^2-K_{ij}K^{ij}=0$, so the Hamiltonian
constraint is precisely~\eqref{eq:polar-initial-mass}.

Use the same static outer collar $r^{\mathsf M}_{\rm out}$ and cutoff
as in~\eqref{eq:collar-mass-matching}, taking
$r=r_*+\chi(r^{\mathsf M}_{\rm out}-r_*)$ and
$v(R)=\varepsilon f((R-R_0)/\ell)$, where $f$ is a fixed smooth compact
bump with a falling edge. Choose its support disjoint from the collar.
At $v=0$, Eq.~\eqref{eq:polar-initial-mass} reduces to the static mass
functional. Its first variation at equilibrium vanishes for every
embedding variation; hence the derivative of
$m(B;\mathsf M,v)-\mathsf M$ with respect to $\mathsf M$ is $-1$.
The implicit function theorem supplies the matching mass. The source
change is $O(v^2)$ on an interval of length $O(\ell)$, so ordinary ODE
dependence gives $\mathsf M-M_*=O(\varepsilon^2\ell)$ uniformly in the
pulse width. Equivalently, the endpoint map is smooth in the uniform
velocity norm; it uses no derivatives of $v$. The matching solution is
static near both faces by ODE uniqueness, and therefore satisfies the
geometric and traction compatibility conditions to every order.

Inside the pulse, $r=r_*$, while $m_R$ and $\eta_R$ remain bounded.
Differentiating the density in~\eqref{eq:polar-initial-mass} yields
\begin{equation}
 \partial_r U_+
 =\frac{\gamma_v^2(1-cv)}{r_R}\,v_R+O(1).
 \label{eq:shell-initial-compression}
\end{equation}
The remainder is bounded independently of $\ell$. The positive coefficient
of $v_R$ proves the initial compression asserted in the theorem.

For completeness, the continuation argument uses state bounds before it
uses derivative bounds. Equations~\eqref{eq:shell-characteristic-sources}
along their two characteristics, and the bounded equations for $z$, keep
the values in a fixed compact neighborhood for a short time independent
of $\ell$. Choose a wider interior interval around the pulse. Its distance
from either face and the bounded light speed $\alpha/a$ give a positive
travel time independent of $\ell$. The complementary initial regions
have uniform smooth bounds; near each face they are exactly static until
an incoming signal arrives. The distant pulse can change the lapse there
by a spatially constant factor depending on time. This is a time
reparametrization of the static collar and produces no material forcing.

Smooth local evolution follows from the radial hyperbolic system and the
compatible free-face construction. If the state stays admissible and
$\|\partial_r(\nu,v)\|_\infty$ stays bounded on a finite time interval,
differentiating the system gives linear estimates for each higher Sobolev
norm with coefficients controlled by these first derivatives. The radial
constraints recover the metric and label derivatives without loss; the
free-face energy has zero traction work as in
Appendix~\ref{app:radial}. Thus the solution continues unless a first state
derivative becomes unbounded. Apply the Riccati comparison before the
uniform state and travel times. Any earlier breakdown has the same
continuation obstruction, and the uniformly regular complementary regions
place it in the material interior. This proves
Theorem~\ref{thm:shell-steepening}. Entropy-admissible weak evolution
beyond $T_*$ is a separate problem.

\section*{Data availability}
The numerical examples use no external observational or experimental data.
The constitutive law, field equations, boundary conditions, parameters,
and derivations needed to reproduce them are given in the article.

\acknowledgments
This work is financially supported by VinUniversity under the Environmental
Intelligence (CEI) Grant (No. VUNI.CEI.FS 0009).

\end{document}